\documentclass{article}

\usepackage{arxiv}

\usepackage[utf8]{inputenc} 
\usepackage[T1]{fontenc}    
\usepackage{hyperref}       
\usepackage{url}            
\usepackage{booktabs}       
\usepackage{amsfonts}       
\usepackage{nicefrac}       
\usepackage{microtype}      
\usepackage{graphicx}
\usepackage[numbers,sort&compress]{natbib}
\usepackage{doi}

\usepackage{changepage}

\usepackage{etoolbox}

\usepackage{gensymb}

\usepackage{amssymb}
\usepackage{amsthm}
\usepackage{amsmath}
\usepackage{enumitem}
\usepackage{textcomp}

\newtheorem{theorem}{Theorem}[section]

\newtheorem{corollary}[theorem]{Corollary}

\newtheorem{remark}{Remark}[section]

\usepackage{cleveref}
\usepackage{hyperref,xcolor}
\hypersetup{
colorlinks,
linkcolor={blue!90!black},
citecolor={blue!90!black},
urlcolor={blue!80!black}
}

\usepackage{jabbrv}

\usepackage{algorithm}
\usepackage{algorithmicx}

\usepackage{lineno} 

\usepackage[font=small]{caption}

\title{A Fast Diagonalization Algorithm to Enable Singular Value Decomposition of Large Matrices for Efficient Template Matching }

\date{} 					

\author{
  \begin{tabular}{ccc}
    \href{https://orcid.org/0000-0002-0924-4410}{\includegraphics[scale=0.06]{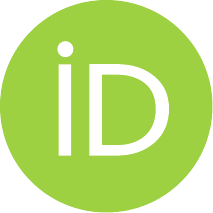}\hspace{1mm}Matthew Giammar}\thanks{\tiny conceptualization, methodology, formal analysis, investigation, data curation, software, visualization, writing -original draft, writing -review and editing.} &
    \href{https://orcid.org/0000-0001-9162-0421}{\includegraphics[scale=0.06]{orcid.pdf}\hspace{1mm}Bronwyn Lucas}\thanks{\tiny conceptualization, supervision, project administration, funding acquisition, writing - original draft, writing -review and editing.} &
    \href{https://orcid.org/0000-0001-6618-631X}{\includegraphics[scale=0.06]{orcid.pdf}\hspace{1mm} Alexander Strang}\thanks{\tiny conceptualization, validation, methodology, formal analysis, investigation, visualization, supervision, project administration, writing -original draft, writing -review and editing.} \\
    \small \textmd{Center for Computational Biology,} & \small \textmd{Department of Molecular and Cell Biology,} & \small \textmd{Department of Statistics,} \\
    \small \textmd{University of California, Berkeley,} & \small \textmd{Center for Computational Biology,} & \small \textmd{University of California, Berkeley,} \\
    \small \textmd{Berkeley, California, USA} & \small \textmd{University of California, Berkeley,} & \small \textmd{Berkeley, California, USA} \\
     \small \textmd{} & \small \textmd{Berkeley, California, USA} & \small \textmd{alexstrang@berkeley.edu} \\
    & \small \texttt{bronwynlucas@berkeley.edu} & \small \texttt{}
  \end{tabular}
}

\renewcommand{\shorttitle}{Fast Diagonalization for Template Matching}

\hypersetup{
pdftitle={A fast diagonalization algorithm to enable singular value decomposition of large matrices for efficient template matching},
pdfauthor={Matthew Giammar, Bronwyn Lucas, Alexander Strang},
pdfkeywords={Symmetry Exploiting Diagonalization, Low Rank Approximation, Template Matching, Cryogenic electron microscopy},
}

\begin{document}

\maketitle

\vspace{-0.4 in}
\begin{center}
    \textbf{SIGNIFICANCE STATEMENT:} \vspace{0 in}
\end{center}

\begin{adjustwidth}{0.0cm}{0.0cm}
High-resolution template matching can be used to annotate macromolecules in cryogenic electron microscopy images of cells by convolution with a series of templates representing hypothesized molecular structures. However, accurate annotation relies on an exhaustive search of over 20 million orientation-depth combinations per template per image making template matching prohibitively expensive for large-scale analyses, even when using a spatial Fast Fourier Transform. Compressed approximations to the corresponding matrix would accelerate the search and enable annotation of more of the proteome towards a more comprehensive understanding of cellular functions. Here, we introduce a general procedure for efficient diagonalization of symmetric matrices and demonstrate its application on large template matching matrices by exploiting hypothesis symmetries in in-plane rotations. This procedure dramatically reduces the cost to compute and store the singular values and vectors of the matrix  (e.g. compressed by $3.5 \times 10^3$ times at $0.01$\% error) and will make compressed template matching feasible.
\end{adjustwidth}

\vspace{0.0 in}

\begin{center}
    \textbf{ABSTRACT:} \vspace{0 in}
\end{center}

\begin{adjustwidth}{0.0cm}{0.0cm}
Many computational problems possess the following three features: (1) the problem could be solved efficiently if a linear operator could be diagonalized, (2) direct diagonalization is infeasible due to scale, but, (3) the operator is known to commute with a permutation since the problem is (a) symmetric with respect to a change of coordinates, or (b) the operator is block circulant. For example, high-resolution template matching demands repeated convolution of large, high resolution images with an exhaustive list of related templates. The associated linear operator could be compressed, via a low rank approximation, if diagonalized. The required decomposition is too expensive to compute directly, but, the entire problem is symmetric to in-plane rotations. In these cases, the desired eigenfunctions are constrained by the symmetry, and these constraints allow efficient decomposition. We illustrate a parallelized algorithm that allows fast diagonalization of any block-circulant matrix. When the index space can be partitioned into $l$ classes of $m$ interchangeable elements, the algorithm reduces storage costs by $m$, from $\mathcal{O}(m^2l^2)$ to $\mathcal{O}(ml^2)$, and if $w$ workers are available, the algorithm demands $\mathcal{O}(l^2m\log(m)/w) + \mathcal{O}(ml^3/w)$ floating point operations per worker yielding a $m^2$ speedup over direct diagonalization. We demonstrate this procedure to decompose a high precision template matching matrix. We compare runtime on a reduced problem, where the fast procedure ran 205 times faster per recovered feature, and recovered 22.5 times as many features. On a typical template matching matrix, the decomposition required 25 times less time than needed to compute the full-scale matrix. We also demonstrate decomposition of a $\sim$30 times larger template matching matrix that covers the complete set of possible projections represented in a cryo-EM image at 2~Å resolution in 14 minutes. This procedure is more stable and allows recovery of in-plane rotations to arbitrary precision. We recommend it for efficient, multi-precision template matching searches.
\end{adjustwidth}



\newgeometry{left=0.5in, right=0.5in, top=0.8in, bottom=0.7in}

\section{Two-Dimensional Template Matching} \label{sec: 2DTM}

Technical advances in cryogenic electron microscopy (cryo-EM) have made it possible to visualize macromolecules in their native, cellular environment. Parallel advances in protein structure prediction are generating unprecedented libraries of atomic models of macromolecular complexes. Together, these advances could yield molecular atlases of cells and lead to an integrated, mechanistic understanding of living systems. However, it is not possible to directly identify specific molecules in cryo-EM images due to molecular crowding and the low signal-to-noise ratio (SNR) resulting from radiation damage.

Detecting the presence of a signal embedded in noisy data is a foundational problem in signal processing. When the expected signal is known and the noise is additive and random, the matched filter provides the optimal test statistic for detecting a weak signal by rejecting the null hypothesis of noise alone \cite{lugt_1964_signal, brunelli_template_1997}. The output of the matched filter is maximal when the hypothesis matches the signal exactly and search algorithms therefore typically require comparing a set of filters sampled across a relevant parameter space. Template matching, in the context of cryo-EM, is a spatial variant of a matched filter in which templates, corresponding to hypothesized image sub-features are located within an image by cross-correlation.

Template matching has been successfully employed to locate specific macromolecules in cryo-EM images by cross-correlation of simulated projections with an experimental image  \cite{rickgauer_single-protein_2017,rickgauer_structural_2024,lucas_locating_2021}. This procedure, alternately called high-resolution template matching or 2-dimensional template matching (2DTM), has been applied to locate ribosomes in \textit{Mycoplasma pneumoniae} \cite{lucas_locating_2021}, distinguish nuclear intermediates of ribosome biogenesis in yeast cell sections \cite{lucas_situ_2022}, characterize ribosome structural dynamics during translation in the human cell periphery \cite{rickgauer_structural_2024} and visualize drug-target interactions \textit{in situ} \cite{lucas_baited_2023}. 2DTM holds the potential to leverage the growing library of predicted structures to visualize entire ``structureomes" in context \cite{lucas_visualizing_2023}.

The extreme molecular crowding within the cell precludes visual identification of particles and high resolution features are required to separate particles from the background \cite{rickgauer_single-protein_2017,dickerson_tilt_2025}. Detecting specific macromolecules in cellular cryo-EM images therefore requires fine grained sampling of template parameters and results in an extreme computational cost. While the positional search of a template in a larger image can be accelerated via the convolution theorem and the Fast Fourier Transform (FFT), the extreme cellular crowding nevertheless demands exhaustive, fine-grained sampling over orientation \cite{rickgauer_single-protein_2017,lucas_locating_2021}. For example, to sample all possible orientations of a 250~Å particle at the full resolution of the data would require 0.46 degree increments, or $\sim$150 million projections \cite{crowther_reconstruction_1970}.  In practice, the search is performed in two steps with an initial coarser search using 1.6 million orientations  \cite{giammar_leopard-em_2026}. To recover as many particles as possible and to capture 3D information from a 2D image, each one of these orientations is additionally assessed over 13 distinct depths within our sample producing over 20 million distinct orientation-depth combinations \cite{rickgauer_single-protein_2017,giammar_leopard-em_2026}. Parallelization using GPUs has increased throughput \cite{lucas_locating_2021}, but even with modern hardware, annotating an individual molecular species in a typical cryo-EM image ($4096\times4096$ pixels) takes \>4 hours \cite{giammar_leopard-em_2026} limiting its broader application to structural biology.

The most expensive step in the search procedure in 2DTM is cross-correlation which is equivalent to performing a large matrix-vector product. Each row of the matrix represents a hypothesis regarding the particle orientation-depth combination and columns represent a matched filter pixel intensity. We call this operator, which produces cross-correlation values from an image, the \textit{template matching matrix}.
Exhaustive orientational sampling quickly explodes the size of this matrix.

A low rank approximation could accelerate a search by expressing the template matching matrix as a product of smaller matrices. Approximating to rank $r$ reduces the cost of a product with a $m\times n$ matrix from $\mathcal{O}(mn)$ to $\mathcal{O}((m+n)r)$, so provides fast multiplication when $r\ll \text{min}(m, n)$. The utility of low-rank approximations for analysis of cryo-EM images has previously been demonstrated in the classification of molecular states by averaging over statistical noise \cite{vanheel_1981_multivariate, roseman2004findem, heel_multivariate_2016}. The truncated singular value decomposition (SVD) gives the optimal such approximation under any unitary matrix norm, for every fixed rank \cite{eckart1936approximation, mirsky1960symmetric}. After whitening, the SVD also minimizes, for a fixed rank, the expected square error in an approximated cross-correlogram (CCG) produced by applying the low rank approximation to a randomly sampled image (see Appendix \ref{app: optimality of SVD}). Since template matching searches declare detections based on cross-correlogram values, minimizing the expected square error in the CCG minimizes the errors passed downstream to decision making. 

While optimal, the SVD cannot be recovered since full-scale \textit{template matching matrices} ($20 \times 10^6$ rows, $512^2$ columns) are far too large to diagonalize directly. Consistently, a particle detection strategy similar to 2DTM leveraging filter banks of eigenimages has been presented \cite{sigworth_classical_2004}, but has not been adopted in practice because calculating the SVD becomes prohibitively expensive.

While large, the template matching matrix is structured. The collection of hypotheses treat orientations uniformly. When considering orientations alone they only distinguish in-plane rotations with respect to an arbitrarily chosen reference direction used to align the image and template. Each in-plane rotation of the template relative to the image could be achieved by an opposing rotation of the image relative to the template. Therefore, products between the template matching matrix and images should be invariant to the arbitrary choice of reference in-plane rotation. This invariance to in-plane rotations introduces a $SO(2)$ symmetry in the template matching matrix with respect to the projection plane. 

These symmetries could be exploited to accelerate decomposition \cite{diaconis1990efficient}. In particular, \emph{any} linear operator whose action is symmetric with respect to an arbitrary bijective transformation of its index space admits a fast, parallel diagonalization procedure. Symmetries of this kind occur naturally when the coordinates used to express the index space require arbitrary choices (e.g.~selection of an origin).

This paper reviews the fundamental mathematics needed to exploit symmetries when diagonalizing, then applies the associated algorithms to the large template matching matrices required for in-situ cryo-EM. We present the general theory and provide references to parallel sources in Section~\ref{sec: Fast Diagonalization}. We show that the associated eigenfunctions are steerable \cite{freeman1990steerable} under the transformation associated with the symmetry (e.g.~rotation), so admit a separable form related to Fourier series in a coordinate system where the transformation is a translation \cite{Hel-Or_1998_canonical}. We then illustrate an fast algorithm that exploits this separable form.

We use this \emph{symmetry exploiting SVD} procedure to decompose large template matching matrices in Section \ref{sec: efficient template matching}. We show its advantage over direct decomposition in terms of feasibility, runtime, scalability, stability, and reconstruction accuracy. We discuss applications of the resulting low rank approximations for multi-precision search, as well as extensions to larger matrices that represent more hypotheses, in Section \ref{sec: discussion}.


\section{Fast Diagonalization} \label{sec: Fast Diagonalization}



To identify a general problem, we will work outwards from the template matching context.  

\begin{enumerate}
    \item We seek a singular value decomposition of a rectangular matrix. The singular value decomposition of a rectangular matrix $M$ can be recovered from the eigenvalue decomposition of a related square matrix ($M^* M$, $M M^*$, or the block matrix $[0, M ; M^*, 0]$). Accordingly, we pursue a general eigenvalue decomposition algorithm. 
    \item The template matching matrix discretizes an integral operator that treats all orientation and position coordinates continuously. We pursue a general eigenfunction decomposition that applies to linear operators so that our theory addresses different discretizations consistently. 
    \item The template matching matrix exhibits a symmetry derived from an ambiguity in the choice of coordinates used to describe the imaging plane. In-plane rotations of the template can be interchanged with in-plane rotations of the image. This is a symmetry in the action of the linear operator. The cross-correlelogram produced by applying the operator to a rotated image equals the rotated cross-correlelogram produced by applying the operator to an unrotated image. Accordingly, we pursue an eigenvalue decomposition for linear operators that commute with bijective transformations of the underlying index space. 
\end{enumerate}

The spectra of linear operators that are symmetric under a group action is a well studied topic in harmonic analysis \cite{folland2016course}. 
The algorithm we suggest is based on Davis' classic analysis of block circulant matrices \cite{davis1979circulant}. For a modern review see \cite{olson2014circulant}. For applications in mechanics, see \cite{dickens1992modal,thomas1974standing,thomas1979dynamics,tran2001component,tran2009component,williams1986algorithm,williams1986algorithm}. For a similar treatment in statistics, see \cite{diaconis1988group} and \cite{diaconis1990efficient}. 

The same algorithm has been proposed for matched filtering against smooth deformations of a filter, first against in-plane rotations and scalings of a single filter \cite{perona1995deformable,uenohara1998optimal}, where the procedure simplifies, and later against in-plane rotations of multiple filters \cite{hilai1994recognition,jogan2003karhunen,vonesch_2013_design}. Similar procedures have been proposed in single particle cryo-EM for fast, rotation invariant, principal component analysis of ensembles of experimental microscopy images \cite{ponce_2011_computing,zhao2013fourier,zhao_2016_fast}.

Consider a linear operator $M$ that maps from functions $v: \mathcal{X} \rightarrow \mathbb{C}$ to functions $w: \mathcal{X} \rightarrow \mathbb{C}$, where $\mathcal{X}$ is an underlying index space. If $\mathcal{X}$ is a finite set, then we can enumerate its elements, and $M$ may be represented with a square matrix. Alternately, if $\mathcal{X}$ is a subset of $\mathbb{R}^d$, then $M$ may represent an integral or differential operator. 

A function $v: \mathcal{X} \rightarrow \mathbb{C}$ is an eigenfunction of $M$ if $v \neq 0$ and:
\begin{align}
    M[v](x) = \lambda v(x)
\end{align}
for some eigenvalue $\lambda \in \mathbb{C}$. We aim to find all of the eigenpairs of operators $M$ that admit a symmetry of action. 

A linear operator admits a symmetry of action if it commutes with a bijective transformation $\sigma: \mathcal{X} \rightarrow \mathcal{X}$:
\begin{align} \label{eqn: generic symmetry of action}
    M[v \circ \sigma] = M[v] \circ \sigma.
\end{align}

For example, if $M$ is the differential operator that evaluates a derivative, and $\mathcal{X} = \mathbb{R}$, then $M$ commutes with all translations, $\frac{d}{dx} [v(x - s)] = v'(x - s)$. In the template matching setting, $M$ is constructed from a template matching matrix that approximates an integral operator, $\mathcal{X}$ is the imaging plane, $v$ are images, and $\sigma$ is any in-plane rotation.

If $\mathcal{X}$ is a finite set, then $\sigma$ is a permutation. Let $P_{\sigma}$ denote the corresponding permutation matrix. Then $M$ admits a symmetry of action with respect to $\sigma$ if and only if it commutes with the permutation matrix:
\begin{align} \label{eqn: commute with permutation}
    M P_{\sigma} = P_{\sigma} M \quad \Rightarrow \quad M = P_{\sigma}^{-1} M P_{\sigma} = P_{\sigma}^{\intercal} M P_{\sigma}. 
\end{align}
Multiplying by a permutation on the right, and its transpose on the left, re-orders the rows and columns by rearranging their indices. So, a square matrix admits a symmetry of action if and only if it is invariant to some matched reordering of the column and row indices.  

Symmetries of this kind occur across diverse applications. For example, if $M$ is an adjacency matrix, or Laplacian, defined on a graph $\mathcal{G}$, and some of the nodes in $\mathcal{G}$ are indistinguishable by topology alone, then $M$ will admit a symmetry of action with respect to permutations of those node labels \cite{godsil2013algebraic}. Alternately, if $M$ is a covariance matrix for a collection of random variables whose joint distribution is invariant to some permutations of the variable indices, then $M$ will admit a symmetry of action with respect to permutations of those indices.

Matrices that commute with permutations can be diagonalized efficiently since any pair of diagonalizable matrices that commute are simultaneously diagonalizable \cite{hom1985matrix}. Therefore, the eigenvectors of $M$ are constrained to the eigenspaces of $P_{\sigma}$. Since the possible eigenvectors of any permutation matrix are known analytically, recognizing the symmetry dramatically reduces the degrees of freedom in the decomposition of $M$ that must be computed. If the eigenvalues of $P_{\sigma}$ are all distinct, then its eigenvectors are uniquely defined, so the eigenvectors of $M$ are fully determined by the eigenvectors of $P_{\sigma}$. If some of the eigenvalues of $P_{\sigma}$ repeat, then some of the eigenspaces of $P_{\sigma}$ are multidimensional, so the eigenvectors of $M$ are not uniquely defined by $P_{\sigma}$ alone. In this case, the eigenvectors of $M$ are restricted to linear combinations of eigenvectors of $P_{\sigma}$ with matching eigenvalues.  The number of remaining degrees of freedom in the eigenvectors of $M$ equals the sum, over all repeated eigenvalues of $P_{\sigma}$, of the geometric multiplicity of each eigenvalue minus one.

Restricting the eigenvectors of $M$ to combinations of eigenvectors of $P_{\sigma}$ constrains the eigenvectors of $M$ to the discretization of functions that are ``steerable" under $\sigma$. A steerable function space is a finite dimensional subspace of functions that is closed under a family of smooth deformations (e.g. translations, rotations, scalings) \cite{freeman1990steerable}. All deformations of steerable functions may be represented by varying the coefficients in an expansion of the function on a fixed basis \cite{Hel-Or_1998_canonical}. All eigenvectors of $P_{\sigma}$ are steerable under the action of the permutation since, by definition, they must remain proportional to themselves after permutation. Any linear combination of eigenvectors of $P_{\sigma}$ that share the same eigenvalue is also an eigenvector of $P_{\sigma}$. Therefore, all of the eigenvectors of $M$ must represent functions that are steerable under the transformation $\sigma$. 

The space of steerable functions associated with a family of smooth deformations may be derived from the associated Lie group \cite{Hel-Or_1998_canonical,teo1998lie,michaelis1995lie}. A Lie group is a collection of parameterized deformations that satisfies the standard group conditions and whose inverse and composition vary smoothly in the deformation parameters \cite{cohen1911introduction}. The generator of the Lie group is the differential operator that performs infinitesimal deformations \cite{Hel-Or_1998_canonical,teo1998lie}. Standard deformations have simple generators. For instance, the generators associated with translation, rotation, and scaling are the partial derivative in the direction of translation, the angular partial derivative, and the radial logarithmic derivative. Non-infinitesimal deformations may be represented by exponentiating the generator, or applying it iteratively via a series expansion. Function spaces that are steerable under the deformation must remain closed under arbitrary group actions, so must remain closed under the generator, its iteration, and its exponential. Setting a function proportional to itself under the action of the generator produces a differential equation whose fundamental solutions span the space of steerable functions \cite{Hel-Or_1998_canonical}. For instance, if $f(x)$ is steerable under translations, then $f(x) \propto - \frac{d}{dx} f(x)$ for some $\lambda$, so $f(x) = f(0) e^{-\gamma x}$ for some constant $\gamma$. 

The solution for translations extends to all single parameter Lie groups since every single parameter Lie group is equivalent to a translation after a change of coordinates \cite{cohen1911introduction}. For example, rotations are translations in the angular coordinate if we represent a plane using polar coordinates. As a result, the steerable functions associated with any single parameter Lie group are exponential functions in the coordinate that the deformation translates \cite{Hel-Or_1998_canonical}. If translation in that coordinate acts periodically, then the coefficient of the exponential must be imaginary, so the associated exponential functions are Fourier basis functions. It will follow that the features produced by a SVD of a template matching matrix should be separable functions, constructed as the product of a radial function with an angular Fourier mode. The same observations apply to multi-family Abelian Lie groups (e.g. translation in multiple dimensions) \cite{segman1992canonical}, but do not apply for non-Abelian groups that cannot be converted to translation operations in a new coordinate system (e.g. rotations in three-dimensions) \cite{Hel-Or_1998_canonical}. 

To derive these results for arbitrary bijections, $\sigma$, we will adopt a coordinate system for the index space $\mathcal{X}$ in which $\sigma$ acts as a translation, then iteratively apply $\sigma$ to constrain the space of steerable functions. 
As in \cite{diaconis1988group}, the appropriate coordinate system is produced by the orbit decomposition of $\mathcal{X}$ under the cyclic group generated by  $\sigma$ \cite{dummit2003abstract,serre1977linear}. The orbit decomposition is constructed in three steps:

\begin{enumerate}
    \item Define an equivalence relation $x \sim x'$ if $x' = \sigma^n(x)$ or $x' = \sigma^{-n}(x)$ for some $n$. Then $x$ and $x'$ are equivalent if it is possible to convert one into the other by repeated application of the permutation.
    \item Let $\mathcal{X} \setminus \sigma$ denote the quotient space of $\mathcal{X}$ with respect to $\sim$. It contains the equivalence classes created by $\sigma$. Let $E(x) \in \mathcal{X} \setminus \sigma$ denote the equivalence class containing an element of the index space, $x$. 
    \item Assign each equivalence class $E \in \mathcal{X} \setminus \sigma$ a representative element $x_0(E)$. Then, index the elements of $E$ by the number of applications of the permutation needed to convert $x_0(E)$ to $x$; $j(x) = n$ if $x = \sigma^n(x_0(E(x)))$.
\end{enumerate}

If $\mathcal{X}$ is finite, then the orbit decomposition assigns each element $x \in \mathcal{X}$ a pair of indices that can be arranged lexicographically. Each $x$ is assigned an outer index representing its equivalence class, and an inner index representing its position in its equivalence class, $x \rightarrow (E(x),j(x))$. In these coordinates, every permutation acts as translation in a periodic domain, $\sigma^n(x) \rightarrow (E(x), j(x) + n)$ where $j(x) + n$ is meant modulo $|E(x)|$ \cite{dummit2003abstract,serre1977linear}. Extending this argument to smooth single-parameter deformations by considering a limit of infinitesimal deformation recovers the coordinate system where the action of the Lie group is equivalent to translation \cite{cohen1911introduction}.

In our applied problem, the template matching matrix admits a symmetry of action with respect to in-plane rotations. In-plane rotations exchange any pair of points that are equidistant from the origin. So, $\mathcal{X}$ is the image plane, the equivalence classes are concentric circles centered at the origin, and the quotient space $\mathcal{X} \setminus \sigma$ is indexed by radii. If we set $x_0(E) = [\|x\|,0]$, then position within equivalence class is indexed by angle counter-clockwise from the horizontal axis. Thus, the orbit decomposition expresses every point in the plane using polar coordinates: $(E(x),j(x)) \leftrightarrow (\rho(x),\psi(x)).$ In-plane rotations are represented by translation in the angular coordinate. 

If $M$ commutes with $P_{\sigma}$, then ordering its indices lexicographically according to the orbit decomposition reveals its structure. Arrange its columns and rows into blocks indexed by a pair of equivalence classes $E, E'$. Let $M^{(E,E')}$ denote the matching block. Then, applying $P_{\sigma}$ to the left and right sends $M^{(E,E')}_{i,j}$ to $M^{(E,E')}_{i + 1, j + 1}$. So, to maintain symmetry:
\begin{align} \label{eqn: block circulant}
    M^{(E,E')}_{i,j} = M^{(E,E')}_{i+n,j+n}
\end{align}
for all $n$. It follows that $M^{(E,E')}$ is a circulant matrix \cite{davis1979circulant,gray2006toeplitz} and $M$ is block circulant \cite{olson2014circulant}. 

The eigenvectors of $M$ may also be expressed elegantly in these coordinates.


\begin{theorem} \label{thm: separable eigenfunctions}
     Suppose that $M$ is a linear operator that maps functions on $\mathcal{X}$ to functions on $\mathcal{X}$, and $M$ admits a symmetry of action with respect to a bijective mapping $\sigma: \mathcal{X} \rightarrow \mathcal{X}$. If $v$ is an eigenfunction of $M$ with a simple eigenvalue $\lambda$ and $v$ is selected so that $|v(x)|$ does not diverge for any sequence of inputs $x$, then $v$ must be a separable function of the form:
    \begin{align} \label{eqn: separable form}
        v(x) = a(E(x)) \times e^{i 2 \pi \omega j(x)}
    \end{align}
    for some function $a: \mathcal{X} \setminus \sigma \rightarrow \mathbb{C}$ and some scalar $\omega \in \mathbb{R}$. 
\end{theorem}

\begin{remark}
The separable form provided in \eqref{eqn: separable form} is steerable under the action of the permutation. To permute $n$ times, multiply by $e^{i 2 \pi \omega n}$. 
\end{remark}

\begin{remark} In the template matching setting $\mathcal{X}$ is $\mathbb{R}^2$, and $\sigma$ is any in-plane rotation. Accordingly, the orbit decomposition expresses $\mathbb{R}^2$ in polar coordinates, where equivalence classes are indexed by radius, $\rho$, and position within equivalence class are indexed by in-plane angle, $\psi$. Then, Theorem \ref{thm: separable eigenfunctions} implies that all eigenfunctions of an integral operator that commutes with $\sigma$ must take the separable form:
\begin{equation}
    v(\rho,\psi) = a(\rho) e^{i k \psi}
\end{equation}
where $a(\rho)$ is a complex valued function that assigns amplitudes to radii, $k$ is an integer wave-number, and $\psi$ is expressed in radians. The wavenumber $k$ is integer valued to ensure periodicity. So, to diagonalize a template matching operator, we only need recover the radial amplitude functions, $a(\rho)$. Past work in cryo-EM has observed and exploited this structure for rotation invariant principal component analyses of single particle images (c.f.~\cite{sigworth_classical_2004, zhao_2014_rotationally,ponce_2011_computing}).
\end{remark}

\begin{corollary} \label{cor: separable form discrete}
    If the assumptions of Theorem \ref{thm: separable eigenfunctions} hold, and $\mathcal{X}$ is finite, then:
    \begin{align} \label{eqn: separable form finite case}
        v(x) = a(E(x)) \times \phi(j(x);k,|E(x)|), \quad \text{ where } \quad \phi(j;k,n) = \tfrac{1}{\sqrt{n}}e^{i 2 \pi \tfrac{k}{n} j}
    \end{align}
    is the discrete Fourier transform (DFT) basis function on vectors of length $n$ with wavenumber $k \in [0,n-1]$. 
\end{corollary}

\begin{proof}

Suppose that $v(x)$ is an eigenfunction of $M$. Then $\lambda v = M[v] $. Since $M$ commutes with $\sigma$, it must also be true that $M[v \circ \sigma] =  M[v] \circ \sigma = (\lambda v) \circ \sigma = \lambda (v \circ \sigma)$. Recursing the argument, $v \circ \sigma^n$ must be an eigenfunction of $M$ with eigenvalue $\lambda$ for all $n$.

If $\lambda$ is simple, then the associated eigenspace is one-dimensional \cite{strang2012linear}. So, it must be true that $v \circ \sigma^n (x) \propto v(x) $ for each $n$. Therefore, every simple eigenfunction must be proportional to itself under application of the transform $\sigma$, e.g.~every simple eigenfunction is steerable.

If we express $x \rightarrow (E(x),j(x))$, then $v \circ \sigma^n (x) = v(E(x), j(x) + n)$. Therefore, when $\lambda$ is simple:
\begin{align} \label{eqn: translation within equivalence class}
v(E,j + n) \propto v(E,j) 
\end{align}
for all equivalence classes $E$, indices $j$, and iteration counts $n$. In other words, restricted to each equivalence class, $v$ must be proportional to itself under all translations. 

The only functions that are proportional to themselves under arbitrary translations are exponential functions \cite{aczel1989functional,aczel2006lectures,hewitt1965real,folland2016course}.  To recover the exponential, set $a(E) = v(x_0(E)) = v(E,0)$, then recurse. Since $v$ is proportional to itself under application of the transform, $v(E,j+1) = v(E,j) \times \gamma$ for some $\gamma$. Therefore $v(E,j) = a(E) \times \gamma^j$.
%

If $|E|$ is finite, then $\gamma$ must be chosen so that $v(E,j)$ is a periodic function of $j$. Then $|\gamma| = 1$. If $|E|$ is infinite, and $v(x)$ does not diverge for some limiting sequence of $x$, then $|\gamma|$ must also equal 1. In either case, $\gamma = \exp(i 2 \pi \omega)$ for some complex phase $\omega$. It follows that:
\begin{align}
    v(E,j) = a(E) e^{i 2 \pi \omega j}
\end{align}%
for some amplitude function $a$, and some frequency $\omega$. 

If $|\mathcal{X}|$ is finite, then every equivalence class is finite, so periodicity requires that $|E| \times \omega$ is an integer for every $|E|$. It follows that, for each equivalence class, $\omega =  k(E)/|E|$ for some integer wavenumber $k(E)$ that could vary by equivalence class. \end{proof}

\begin{remark}
\eqref{eqn: separable form} expresses each eigenfunction in a separable form. Each eigenfunction may be expressed as a product of an outer function that assigns each equivalence class a complex amplitude, and an inner function that behaves like a Fourier basis function within each equivalence class. 

Each eigenfunction $v$ has $|\mathcal{X}\setminus \sigma| + 1$ degrees of freedom: one frequency and one amplitude per equivalence class. The frequencies are fixed by the sizes of the equivalence classes, so are inherited from the symmetry. The amplitudes assigned to each equivalence, $a$,  represent the possible linear combinations of eigenvectors of $P_{\sigma}$ that share an eigenvalue. These must be computed from the entries of the matrix $M$.

In the finite case, every eigenfunction must be periodic within each equivalence class. Since the frequency is shared across all equivalence classes, the available frequencies must be compatible with $|E|$ for each $E$ such that $|a(E)| \neq 0$.
\end{remark}

So, to recover the eigenvalue decomposition of $M$, we need only recover the list of eigenvalues and amplitudes. These may be recovered by applying a series of transformations to $M$ that reduce it to a sparse matrix, then reorder its entries to convert to a block diagonal form with one block per unique eigenvalue of $P_{\sigma}$. 

Suppose that $M \in \mathbb{C}^{n \times n}$ is a diagonalizable matrix with simple eigenvalues that commutes with a permutation $\sigma$ that divides $\mathcal{X}$ into $l$ equivalence classes of sizes $\{m_j\}_{j=1}^l$. Order the indices of $\mathcal{X}$ lexicographically according to the orbit decomposition. Then $M$ is block-circulant
%
%
with $l^2$ circulant blocks, $M^{(i,j)} \in \mathbb{C}^{m_i \times m_j}$. A circulant matrix is fully specified by its first row. The rows of a $m_i \times m_j$ circulant matrix are periodic with period $p_{ij} =\text{gcd}(m_i,m_j)$ \cite{davis1979circulant,diaconis1988group}.  
So, to store $M$, we need only store the first $p_{ij}$ entries from any block. We will work with these lists directly, not the full block circulant matrix $M$. 

Let $F_{m}$ denote the discrete Fourier transform (DFT) matrix of size $m \times m$ whose columns store the $m$ discrete Fourier basis functions, $[F_m]_{j,k} = \phi(j;k,m)$ \cite{olson2014circulant,strang2012linear,van1992computational}.
%
%
Let
%
%
$F$ denote the block matrix with $l$ diagonal blocks, each equal to a DFT matrix of size $\{m_j\}_{j=1}^l$. 

Circulant matrices are diagonalized by the DFT \cite{davis1979circulant,golub2013matrix,gray2006toeplitz,uenohara1998optimal}. So, let $\hat{M} = F^{-1} M F = \bar{F} M F$. Then $\hat{M}$ has blocks:
\begin{equation}
    \hat{M}^{(i,j)} = \bar{F}_{m_i} M^{(i.j)} F_{m_j}.
\end{equation}

When $m_i = m_j$ every $i,j$ block is square and diagonal, with diagonal entries equal to the DFT of the first row of $M^{(i,j)}$ \cite{olson2014circulant}. Then $\hat{M}$ is a square block matrix with $l^2$ diagonal blocks, each of size $m \times m$. 
If $m_i \neq m_j$, then the $i,j$ block is rectangular, with $p_{ij} = \text{gcd}(m_i,m_j)$ nonzero entries. The nonzero entries do not occur along the main diagonal, but instead run along the diagonal linking the upper left-hand, and lower right-hand corners of the block. They occur at the index pairs that correspond to matched frequencies: $\{(k m_i/p_{ij} + 1, k m_j/p_{ij} + 1)\}_{k=0}^{p_{ij}-1}$ \cite{diaconis1988group,gray2006toeplitz,terras1999fourier}. 
The entries correspond to the DFT of $M^{(i,j)}_{1:p_{ij}}$ scaled by $\sqrt{m_i m_j}/p_{ij}$ \cite{terras1999fourier}. 
In either case, $\hat{M}$ is sparse, with as many nonzero entries as degrees of freedom in $M$. The more entries of $M$ are related by symmetries, the sparser $\hat{M}$. 


 \begin{figure}[t!]
    \centering
    \includegraphics[trim = 120 140 120 100, clip, width = \textwidth]{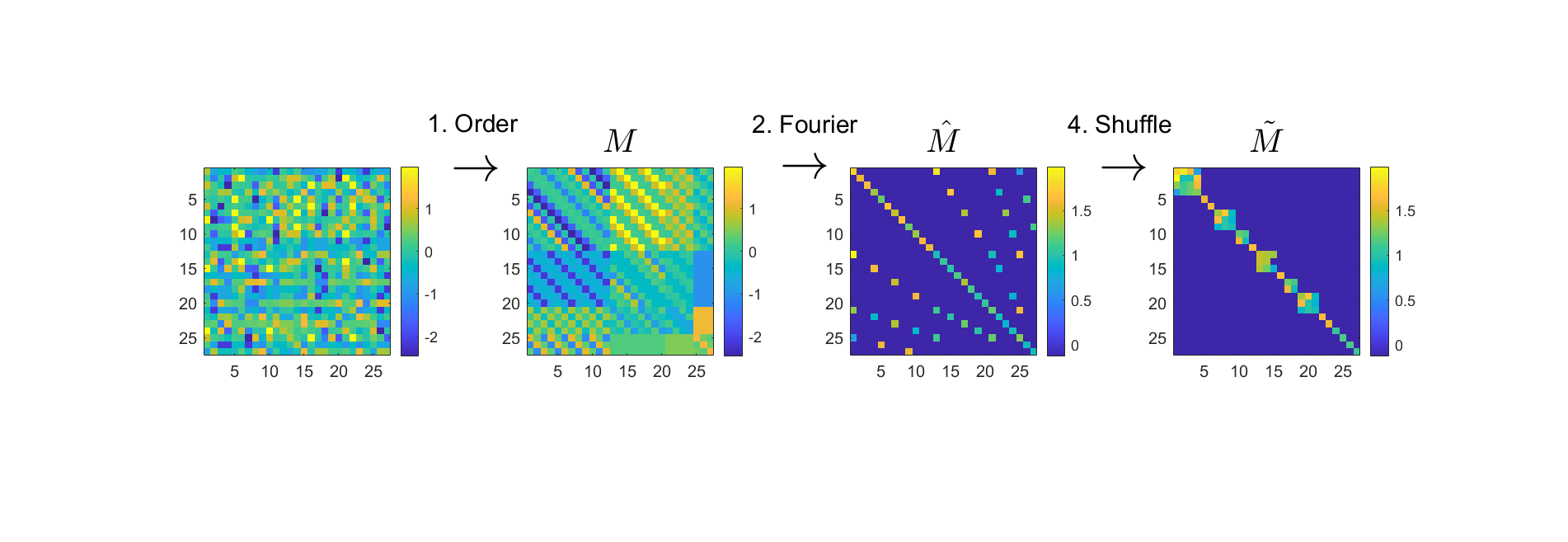}
    \caption{Stages (1.), (2.) and (4.) of Algorithm \ref{alg: variable size} applied to a sample matrix which admits one symmetry of action. The symmetry creates 7 equivalence classes of sizes 10, 8, 4, and 3. 
    \textbf{Left:} The matrix before lexicographical ordering. \textbf{Middle Left:} Ordered lexicographically by equivalence class. The blocks are all circulant. \textbf{Middle Right:} The matrix $\hat{M}$ formed by applying a block DFT to both sides of $M$ (color is absolute value). Dark blue entries are at or well below machine precision. \textbf{Right:} The matrix $\tilde{M}$ after permuting to group all rows and columns that share a nonzero off-diagonal element (color is absolute value). }
   \label{fig: one symmetry shuffled}
\end{figure}

The sparse matrix $\hat{M}$ may be converted into a block diagonal matrix by permuting its rows and columns \cite{olson2014circulant}. The appropriate permutation can be derived by matching frequencies. First, compute the list of frequencies associated with the basis functions represented by the columns of $F$. These are $\{ \{ k/m_j \}_{k=0}^{m_j - 1} \}_{j=1}^l$.
%
%
Sort the list from smallest to largest. Let $S$ denote the matching permutation. Then, group all entries that share the same frequency. Let $b_k$ denote the multiplicity of $k^{th}$ largest unique frequency. 

Let: 
\begin{equation} \tilde{M} = S^{\intercal} \hat{M} S.\end{equation}

Then $\tilde{M}$ is block-diagonal with one block per unique frequency, and with $k^{th}$ block of size $b_k \times b_k$. 


In the special case when all of the equivalence classes have the same number of entries $S$ is a perfect shuffle matrix that exchanges the inner and outer indices of the lexicographical ordering over equivalence class and wave number \cite{d2017shuffling,garsia1985shuffles}. 
Then $\tilde{M}$ is block diagonal with $m$ blocks. Each diagonal block is associated with a wave number and its rows and columns index specific equivalence class pairs.  To recover an entry, set $\tilde{M}^{(k)}_{i,j} = \text{DFT}[M^{(i,j)}_{1,:}](k)$. 

In either case, the matrix $\tilde{M}$ is block diagonal, so can be diagonalized by diagonalizing each of its blocks independently:
\begin{equation}
    \tilde{M}^{(k)} = A^{(k)} \Lambda^{(k)} {A^{(k)}}^{-1}
\end{equation}
where $A^{(k)}$ is the matrix whose columns are the eigenvectors of the $k^{th}$ block, and $\Lambda^{(k)}$ is the diagonal matrix whose diagonal entries are the corresponding eigenvalues.

This analysis suggests a three stage algorithm. Apply a DFT to the first $p_{ij}$ entries of each block of $M$ to recover the nonzero entries of $\hat{M}$. Then, permute the rows and columns of $\hat{M}$ to produce the block diagonal matrix $\tilde{M}$. The necessary permutation depends only on the sizes of the equivalence classes so can be precomputed. Finally, diagonalize each block of $\tilde{M}$. The first and third stages can be fully parallelized, once over the separate equivalence classes, then once over the separate frequency blocks. For equivalent algorithms, restricted to in-plane rotations, see \cite{hilai1994recognition,jogan2003karhunen}. For related, functional, approaches, see \cite{ponce_2011_computing,zhao2013fourier,zhao_2016_fast,vonesch_2013_design}. For equivalent algorithms, restricted to deformations of a single template, see \cite{perona1995deformable,uenohara1998optimal}.  

Algorithm \ref{alg: equal size} enumerates this process when all of the equivalence classes have the same size. For the general form, see Algorithm \ref{alg: variable size} in the appendix. Figure \ref{fig: one symmetry shuffled} illustrates the process for a randomly generated example.

\begin{algorithm}[h!]
\caption{Symmetry Exploiting Diagonalization}
\label{alg: equal size}

Given $M \in \mathbb{C}^{n \times n}$, diagonalizable, block-circulant with $l$ blocks of size $\{m_j\}_{j=1}^l = m$:

\begin{algorithmic}[1] 
\State Order the indices lexicographically with outer indices corresponding to equivalence class and inner indices corresponding to element within an equivalence class.
\State $\hat{m}^{(i,j)} \gets \text{DFT}[M^{(i,j)}_{1,:}]$. Perform in parallel.
\State Set $S$ equal to the perfect shuffle that exchanges $l$ outer and $m$ inner indices.
\State $\tilde{M}^{(k)}_{i,j} \gets \hat{m}^{(i,j)}(k)$ 
\State $(A^{(k)},\Lambda^{(k)}) \gets \text{eig}(\tilde{M}^{(k)})$ for each group $k$. Perform in parallel. 
\end{algorithmic}

\vspace{0.02 in}
Return: $(S, A, \Lambda)$ and $V =  F S A$. Then $(\Lambda, V)$ provide the eigendecomposition $M = V \Lambda V^{-1}$. 
\end{algorithm}






Algorithm \ref{alg: equal size} is significantly cheaper than direct diagonalization since it allows parallelization over both equivalence class and wavenumber while only demanding direct diagonalization of blocks of maximum size $l \times l$. The larger the equivalence classes, the smaller $l$ relative to $n$, so the faster the procedure.

\vspace{0.2 in}
\begin{remark} \label{remark: costs} Suppose that $M \in \mathbb{C}^{n \times n}$ commutes with a permutation $\sigma$ which divides $\mathcal{X}$ into $l$ equivalence classes, all of size $m$.

\begin{enumerate} [leftmargin=0.5cm]
\item \textbf{Storage:} The memory needed to store $M$, diagonalize $M$, and store its eigenvectors is $\mathcal{O}(l^2 m)$.

\item \textbf{Computation:}  The computation needed to diagonalize $M$, given $w$ workers, is $\mathcal{O}(l^2 m \log(m)/w) + \mathcal{O}(m (l - 1)^3/w)$. 
\end{enumerate}

\noindent For the costs when the blocks differ in size, see Remark \ref{remark: costs generic} in the Appendix. 
\end{remark}

We can represent Algorithm \ref{alg: equal size} as a matrix decomposition of the eigenvectors. Each factor in the decomposition represents both a stage of the algorithm and a component of the separable form provided by Theorem \ref{thm: separable eigenfunctions}. 

\begin{theorem}
    Suppose that $M \in \mathbb{C}^{n \times n}$ is a diagonalizable matrix with simple eigenvalues. 
    Let $V$ denote the $n \times n$ matrix whose columns store the eigenvectors of $M$. If $M$ commutes with a permutation $\sigma$ that divides $\mathcal{X}$ into $l$ equivalence classes of sizes $\{m_j\}_{j=1}^l$ then:
    \begin{equation}
        V = F S A
    \end{equation}
    where each factor is $n \times n$ and:
    
    \begin{enumerate}[leftmargin=0.5cm]
        \item $F$ is a block diagonal matrix, with blocks of sizes $\{m_j\}_{j=1}^l$. Each block is $F_{m_j}$ the $m_j \times m_j$ orthonormal DFT matrix. $F$ introduces the oscillatory behavior of the eigenvectors within the equivalence classes. 
        \item $S$ is a permutation matrix that sorts the frequencies in ascending order. $S$ matches complex amplitudes to wave functions. It represents the separable structure of the eigenvectors. 
        \item $A$ is a block diagonal matrix with blocks of size at most $l$. $A$ assigns complex amplitudes to the equivalence classes. 
    \end{enumerate}

    \noindent Computation is only needed to find the entries of $A$ as $F$ and $S$ are determined by the symmetries. 

    %

\end{theorem}




Algorithm \ref{alg: variable size} provides an efficient diagonalization procedures for block circulant matrices. It returns an eigenvalue decomposition. To return a singular value decomposition (SVD), repeat the same procedure to convert from $M$ to $\tilde{M}$, then apply an SVD instead of an eigenvalue decomposition to each frequency block:
\begin{equation}
    \tilde{M}^{(k)} = U^{(k)} \Sigma^{(k)} {W^{(k)}}^{*}.
\end{equation}
where $U^{(k)}$ and $W^{(k)}$ are orthonormal matrices, and $\Sigma^{(k)}$ is diagonal with nonnegative entries. 


\section{... of template matching matrices} \label{sec: efficient template matching}

2DTM is an exhaustive search algorithm which compares simulated 2D projections of a 3D macromolecular structure with an experimental cryo-EM image. A typical search includes 6,602 out-of-plane orientation pairs ($\phi, \theta$) plus 240 in-plane rotations ($\psi$) generated to finely sample $SO(3)$ space \cite{Yershova2010generating}.
Simulated projection images are generated in square boxes, usually $512\times512$ pixels, to model the high resolution information 2DTM depends on for accurate particle annotation \cite{lucas_locating_2021}. Stacking simulated projection images together yields a dense \textit{template matching matrix} $M$ with $n=1.6\times10^6$ rows by $d=262,144$ columns ($512^2$ pixels). We ignore modeling contrast transfer functions (CTFs) for different depth planes in this template matching matrix because practical CTFs include astigmatism fixed by the microscope geometry and therefore do not display in-plane rotational symmetry.

Towards accelerating 2DTM, we seek the singular value decomposition $M = U \Sigma W^*$, where $\Sigma$ is a diagonal matrix storing the singular values, while $U$ and $W$ are orthonormal matrices storing the singular vectors. To distinguish their roles, we will refer to the left singular vectors as hypothesis features, and the right singular vectors as template features. Every simulated projection (row of $M$) may be reconstructed via a linear combination of the template features. Cross-correlations with an image, $x$, may be recapitulated by first computing the inner product $z=(\Sigma W^*)x$ then taking weighed linear combinations of $z$ based on the rows of $U$.
We refer to $U$ as hypothesis features since each row corresponds to a different hypothesized orientation.

Computing a decomposition directly on the full matrix, which we refer to as a \emph{Cartesian SVD}, is computationally intractable given the cubic scaling costs and because the high rank features needed to accurately represent high-resolution, noise-free projections are numerically unstable. 

Representing projection images in polar coordinates introduces a symmetry of action between the in-plane rotation angle of the projection (rows) and the angular coordinate of the image (columns). This symmetry induces a block circulant structure in the template matching matrix. Therefore, Algorithm \ref{alg: equal size} can be applied to reduce the computational costs and increase the numerical stability of a SVD. Crucially, only a single in-plane rotation is needed for each out-of-plane orientation during the decomposition while still allowing analytic recovery of in-plane rotation.


We implemented this \emph{symmetry exploiting polar SVD} procedure for template matching matrices in a Python package and demonstrate its advantage over the \emph{Cartesian SVD} in terms of runtime, numeric stability, and scalability. The package leverages PyTorch for GPU-accelerated numerical operations while also enabling multi-core CPU execution. The process is illustrated in Figure \ref{fig:fig2-pipeline} and represented as a workflow in the Appendix \ref{app: decomp pipeline} (Figure \ref{fig:fig6-workflow}).

\begin{figure}[!b]
    \centering
    \includegraphics[width=\linewidth]{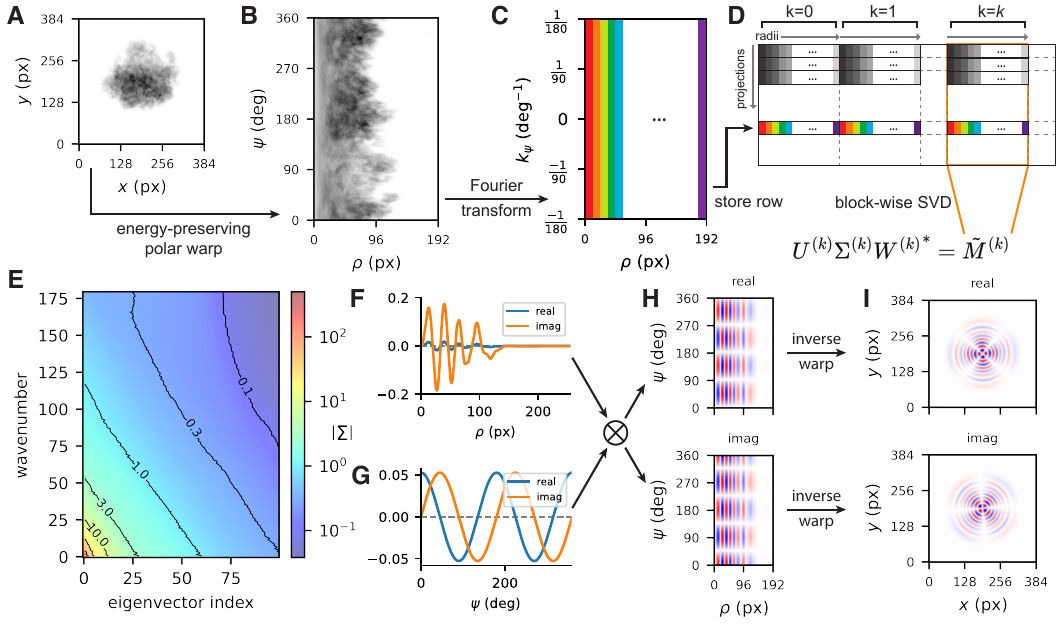}
    \caption{Computational pipeline for applying the symmetry exploiting SVD procedure to simulated projections in a template matching matrix. (A) Single projection of the large ribosomal subunit generated in Cartesian space at a particular orientation with no modeled CTF. (B) The same projection image transformed into standard polar coordinates through an energy-preserving warp. (C) Representation of the polar projection Fourier transformed along the angular dimension. Each color represents the complex amplitudes for a radial ring which are the diagonal entries of a circulant block in the matrix. (D) The full matrix $\tilde{M}$ constructed where each row corresponds to a hypothesized projection and the columns are indexed first by frequency, then radius. Each block of the matrix is decomposed to obtain per-frequency singular values and vectors. (E) Magnitude of extracted singular values indexed by wavenumber and eigenvector index. (F) An eigenvector for a particular angular frequency which corresponds to a complex-valued radial amplitude function. (G) The Fourier mode associated with the eigenvector in F. (H) Taking the outer product of the eigenvector and Fourier mode generates a polar representation of an extracted template feature. (I) Applying the inverse polar warp transformation generates the same template feature in cartesian coordinates. Real and imaginary components are treated independently. }
    \label{fig:fig2-pipeline}
\end{figure}

The \emph{symmetry exploiting SVD} procedure operates in two distinct stages (Figure~\ref{fig:fig2-pipeline}A-C) which are both trivially parallelizable. The first stage simulates projections for all sampled out-of-plane orientations via the Fourier slice theorem, optionally convolves the projections with a user-defined Fourier space filter, and warps the projections from Cartesian to polar coordinates using bi-quintic interpolation (Figure~\ref{fig:fig2-pipeline}B).
A Jacobian correction term based on element area is also applied during the warp transform to maintain a consistent quadrature approximation to the underlying integral operator. We adopt a custom polar coordinate system that preserves the desired rotational symmetries, while achieving asymptotically uniform point density far from the origin, and while controlling the maximum distance between any point in the Cartesian grid and its nearest neighbor in the polar mesh. For details, see Appendix \ref{app: spiral}. The Jacobian correction term and custom mesh reduce the information lost when changing coordinates, and resolve some of the smoothing artifacts observed in \cite{uenohara1998optimal,hilai1994recognition}. Zhao and Singer have proposed a functional approach that does not require conversion onto a new mesh by first by regressing the image values on the Cartesian grid onto the Fourier-Bessel basis (the eigenfunctions of the Laplacian on a disc) which is also a steerable basis  \cite{zhao2013fourier,zhao_2016_fast}.

Once converted into a polar representation, a fast Fourier transform (FFT) is applied along the angular dimension of each projection to obtain the eigenvalues of each circulant block without materializing all in-plane rotations (Figure~\ref{fig:fig2-pipeline}C). These eigenvalues are stored in a intermediary data matrix corresponding to $\tilde{M}$ (Figure~\ref{fig:fig2-pipeline}D). We note the storage step orders columns first by angular frequency, then by radius to emulate the application of the shuffle matrix at no additional cost.

The second stage of the procedure computes the per-angular frequency SVD by grouping the columns of $\tilde{M}$ into blocks indexed by frequency, $\tilde{M}^{(k)}$. These are passed through a Jacobi SVD calculation to compute the complex amplitudes ${W^{(k)}}^*$, singular values $\Sigma^{(k)}$, and associated reconstruction weights $U^{(k)}$. The computation is parallelized across frequency blocks. The singular values are indexed by wave number, $k$, and eigenvector index (Figure~\ref{fig:fig2-pipeline}E). The right singular vector for any singular value is a complex-valued radial amplitude function that assigns a complex amplitude to each circle centered at the origin (Figure~\ref{fig:fig2-pipeline}F). 

To recover the corresponding singular vector of $M$, we compute an outer product between the analytic Fourier mode (Figure~\ref{fig:fig2-pipeline}G) with the radial amplitude function. This produces a complex-valued, polar-space template feature (Figure~\ref{fig:fig2-pipeline}H) in the separable form specified by \eqref{eqn: separable form finite case}. Applying the inverse of the warp transformation from the first stage yields a steerable Cartesian template feature (Figure~\ref{fig:fig2-pipeline}I). These template features are steerable since the response of the template to a rotated copy of the image can be computed directly from the un-rotated response by a complex phase shift.

For real-valued projection image data, complex eigenvalues of each circulant block will have conjugate symmetry between positive and negative angular frequencies. By using the real-valued FFT, we can half the intermediary data storage requirements, half the number of SVD subroutine calls, and half storage costs with no information loss.

\begin{figure}[b!]
    \centering
    \includegraphics[trim = 0 0 0 0, clip, width=0.94\linewidth]{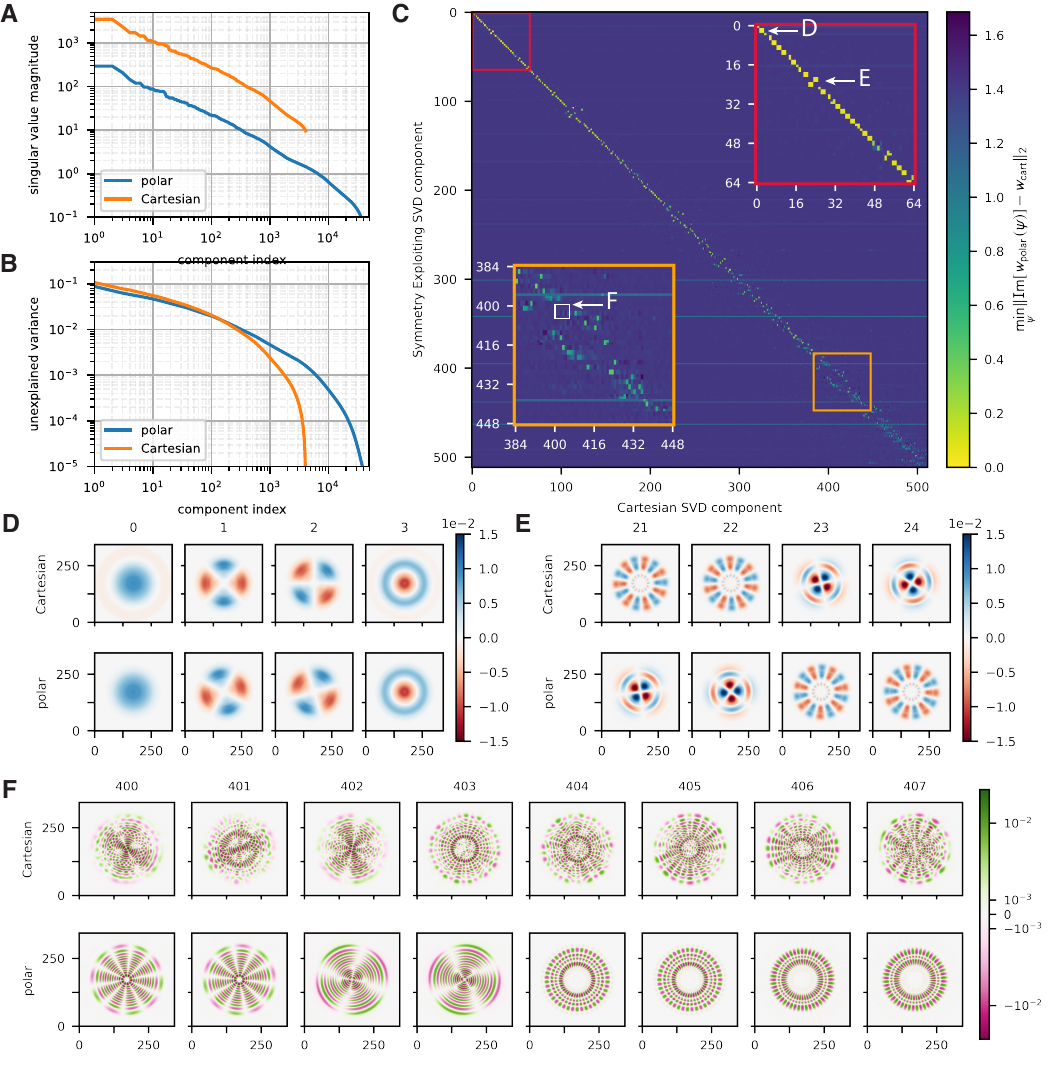}
    \caption{Comparison of singular value and template feature structure from a Cartesian SVD and symmetry exploiting polar SVD procedure applied to the same underlying template matching matrix. (A) Decay of extracted singular values between methods which show similar low rank structure, up to a scaling factor, but diverge at higher components ($r>500$). The Cartesian SVD was only run up to $r=4096$ components due to computational constraints. (B) Comparison of cumulative unexplained variance between both SVD methods. (C) Distance matrix between template features (right singular vectors) for the top 512 singular values. In-plane rotations were applied analytically to polar template features to find the rotation which minimized the $\ell_2$ distance to the Cartesian template features; a consistent rotation was applied across each row. Because \emph{Cartesian SVD} produces real-valued features and because the real and imaginary arguments of an template feature from the \emph{symmetry exploiting polar SVD} for $k_\psi \neq 0$ are a quarter wavelength offset from each other, we compare and display only the complex argument of the polar SVD feature scaled by $\sqrt{2}$ to preserve an image norm of $1$. (D) The top 4 template features display nearly exact registration, up to an in-plane rotation. (E) Example where the same template features show registration but their ordering is swapped. (F) Cases where template feature registration breaks down between the two methods, likely due to numerical instability of the massive Cartesian SVD. Images in F shown on a symmetric log scale for visualization purposes.}
    \label{fig:fig3-comparison}
\end{figure}

We established that the symmetry exploiting decomposition procedure produces consistent results with the \emph{Cartesian SVD} by comparing features extracted from the same underlying template matching matrix for the Large Ribosomal Subunit (LSU). Simulated projection images were generated from the PDB structure 6Q8Y \cite{tesina_structure_2019} at $0.936\text{\AA}/\text{px}$ in $512\times512$ pixel images and with a B-factor scaling of $0.5$ applied to per-atom B-factors deposited in the PDB. Orientations were sampled at 3.5 degree out-of-plane orientation increments and uniformly spaced 2.5 degree in-plane rotation increments using the package \texttt{torch\_so3}, and a cumulative dose filter of $50~ e^{-}/\AA^2$ was applied \cite{grant_2015_measuring}. Templates were cropped to the central $344\times344$ pixels to produce a template matching matrix with $n=485,856$ rows and $d=118,336$ columns. This matrix has one-third as many rows and one-half as many columns as a full-scale template matching matrix. We compared methods on a reduced matrix since the full-scale matrix is too large for the \textit{Cartesian SVD} routine given available computing resources (32 core Intel Xeon Gold 5218 CPU, 1.5TB memory).

The truncated \emph{Cartesian SVD} --- calculated using dask, a parallel computing library in Python \cite{rocklin2015dask} --- took 153 minutes and consumed nearly all memory available to the server to return $r=4096$ features. The \emph{symmetry exploiting polar SVD} procedure, when run on the same CPU-only hardware, was \textbf{9.15} times faster. It required only \textbf{16.3} minutes to obtain over \textbf{22.5} times as many ($> 9\times10^5$) features, so ran \textbf{205} times faster per feature.

The two methods produce a similar sequence of singular values, up to a scaling constant (Figure~\ref{fig:fig3-comparison}A). The scaling reflects the fact that the polar mesh uses more nodes than the Cartesian mesh. The resulting matrix is larger, so distributes some energy to ranks that exceed the size of the Cartesian template matching matrix. Nevertheless, the relative $\ell_2$ error in the reconstructed template matching matrices (cumulative unexplained variance), decay, as a function of rank, at similar rates, until the rank of the approximations approach the dimensions of the matrices (Figure~\ref{fig:fig3-comparison}B). 

The low rank template features extracted by the two methods are similar, though may differ by an arbitrary in-plane rotation (Figure~\ref{fig:fig3-comparison}D), and may differ in order when their singular values are near-equal in magnitude (Figure~\ref{fig:fig3-comparison}E). 

High-rank features diverge as numerical errors degrade the quality of the \emph{Cartesian SVD}, but not the \emph{polar SVD}, and where discrepancies in the polar and Cartesian meshes provide different quadrature approximations to the underlying integral operation (Figure~\ref{fig:fig3-comparison}CF). 
The symmetry exploiting SVD resolves the template features more stably than direct decomposition for two reasons. First, the Fourier modes are resolved analytically, so are not susceptible to numerical error. 
Second, the amplitude functions are recovered via a sequence of $m$ singular value decompositions, one per wavenumber. Since each individual decomposition is smaller than the full decomposition, the associated amplitude function can be computed more stably than the matching singular vector of the full matrix. 

\begin{figure}[t!]
    \centering
    \includegraphics[width=1\linewidth]{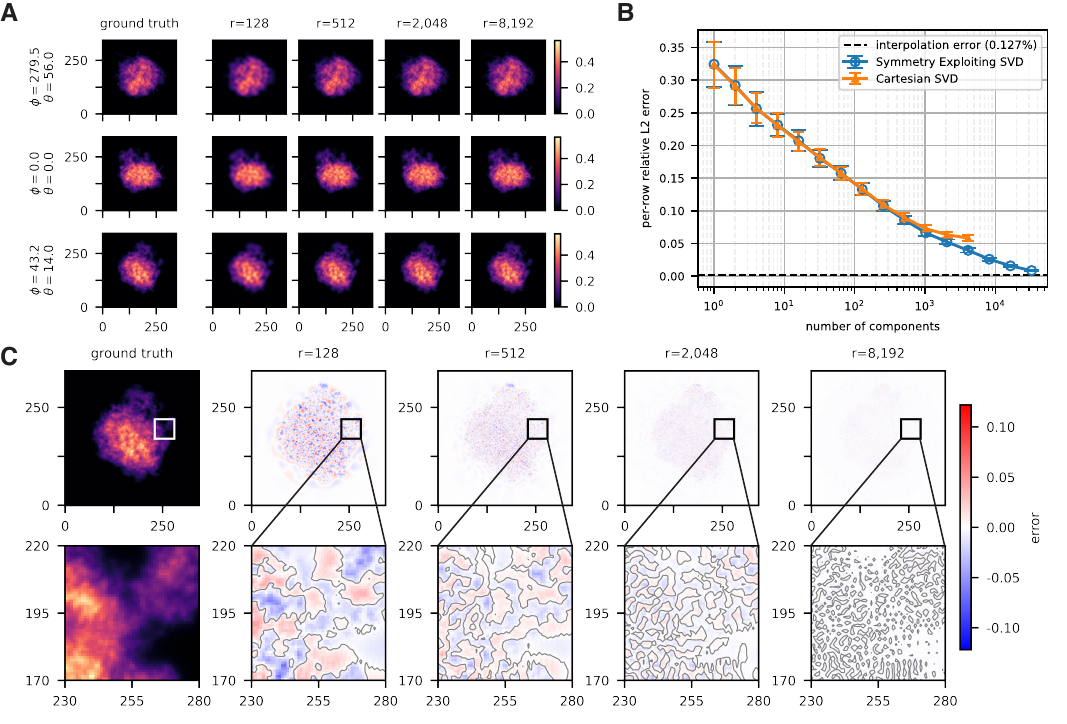}
    \caption{Fidelity of projection reconstruction from the symmetry exploiting SVD procedure for different numbers of components. (A) Reconstructions of three ground-truth projections (rows) using an increasing numbers of components. Out-of-plane orientation angles shown on the left. (B) Comparison of the per-row relative $\ell_2$ reconstruction error between the Cartesian and symmetry exploiting Polar SVD. Polar representations begin outperforming the Cartesian representations past $r\approx1000$. Dashed line is averaged round-trip interpolation error from Polar coordinate transformations applied to the same ground truth projections. (C) Difference between reconstruction ground-truth projection for a single orientation. Top row displays the full projection reconstruction error while the second row displays the same reconstruction error in a zoomed-in region. Contour lines in the second row indicate the zero-crossing in reconstruction error}
    \label{fig:fig4-reconstruction}
\end{figure}

To demonstrate that the features recovered by the \emph{symmetry exploiting polar SVD procedure} provide a valid decomposition, we reconstructed projection images with an increasing number of components. The reconstruction tends towards the ground truth as the number of included components, $r$, increases (Figure~\ref{fig:fig4-reconstruction}A). Prominent visual features are resolved by $r \approx 1000$ with finer details filled in thereafter. We computed the per-row reconstruction error compared to the ground-truth projections for both decomposition methods. Distributions of reconstruction errors are consistent between both methods until $r \approx1000$ where representations from the symmetry exploiting SVD begin to outperform Cartesian SVD representations (Figure~\ref{fig:fig4-reconstruction}B). We note that the reconstruction error for the symmetry exploiting SVD trends towards the round-trip interpolation error for the Cartesian-polar coordinate transformation (0.13\%). The Cartesian SVD reconstructions never surpasses an average reconstruction error of 5\%.

At high component indices, the reconstruction error is consistent across all rows, indicating that we reconstruct each individual projection accurately in addition to the matrix as a whole. Reconstruction errors also shift from low-resolution features in low rank approximations to high-resolution details as more components are included (Figure~\ref{fig:fig4-reconstruction}C).

\begin{figure}[t!]
    \centering
    \includegraphics[width=1\linewidth]{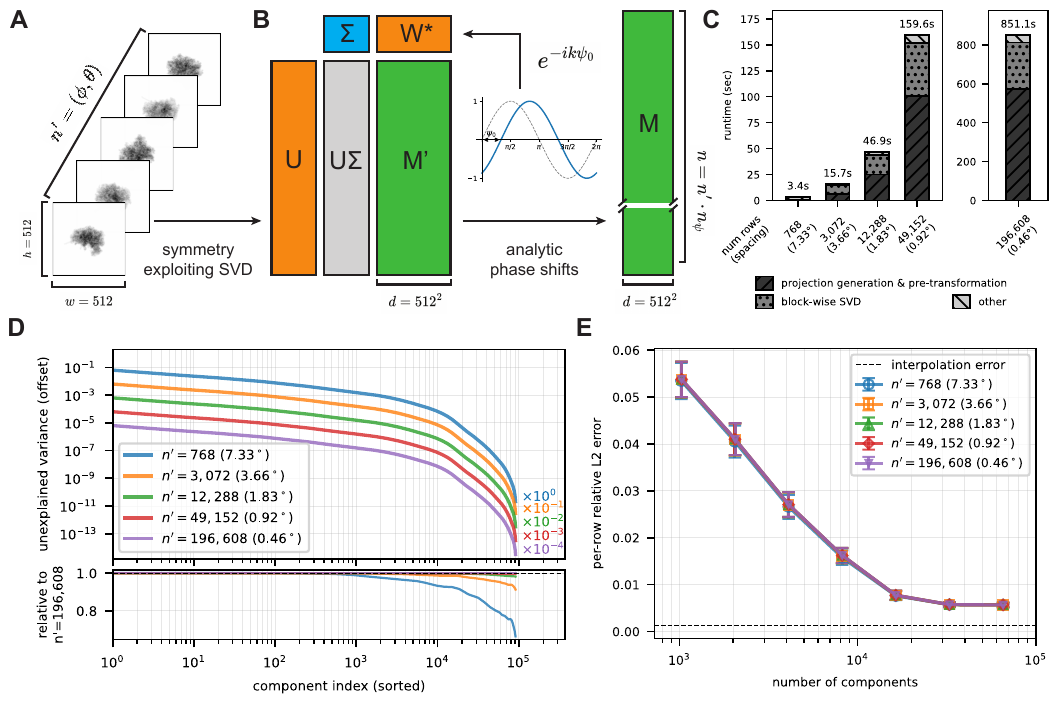}
    \caption{Decomposition of template matching matrices across orientation sampling regimes for the same macromolecular structure. (A) Projections for each sampled out-of-plane orientation are fed into the symmetry exploiting SVD pipeline. $n'$ refers to the number of orientations on the HEALPix grid. (B) Method for analytically recovering in-plane rotations of projections by applying a complex phase shift to the extracted template features which enables recovery of arbitrary in-plane rotation. Assuming discretely sampled in-plane rotations at the same spacing as the largest HEALPix grid ($n_\psi = 782$), then matrix $M$ represents 153 million unique projections. (C) Runtime breakdown of the decomposition pipeline between the projection simulation and pre-transformation stage (Figure~\ref{fig:fig2-pipeline}A-C) and the block-wise decomposition stage (Figure~\ref{fig:fig2-pipeline}D) run on an RTX PRO 6000 Workstation Edition GPU. (D) Cumulative unexplained variance ratios from each independent decomposition. (top) Unexplained variance curves show identical decay. Curves offset for clarity. (bottom) Ratio between the unexplained variance curves, relative to the largest $n'$ result. (E) Reconstruction errors for 768 orientations similar across each sampling regime indicate that for the same number of components each decomposition can reconstruct the ground-truth data to the same fidelity.}
    \label{fig:fig5-fullscale}
\end{figure}

After verifying the symmetry \emph{symmetry exploiting SVD} provides a valid decomposition, we tested how the decomposition pipeline scales with number of hypotheses (rows) in the template matching matrix. We used the HEALPix grid to sample five progressively finer grids on $S^2$ spanning an angular spacing of $7.325^\circ$ to $0.458^\circ$ and between $n'=768$ and $n'=196,608$ out-of-plane orientations, respectively \cite{zonca2019healpy, gorski2005healpix}. For a particle the size of the LSU, this approximates sampling features to a resolution of $\sim$40--2~Å \cite{crowther_reconstruction_1970}. The HEALPix grid is an ideal choice for progressive sampling since each point has equal area and therefore gives uniform quadrature weight to each orientation hypothesis.

Projections of the same LSU structure were generated at varying numbers of out-of-plane orientations in $512\times512$ pixel images again at $0.936\text{\AA}/\text{px}$ and passed through the symmetry exploiting SVD pipeline (Figure~\ref{fig:fig5-fullscale}A). Without directly modeling in-plane rotation, the symmetry class can be recovered by applying an analytic phase shift to the Fourier mode of each template feature (Figure~\ref{fig:fig5-fullscale}B). For a rotation angle $\psi$ and angular mode $k$, this is equivalent to a complex multiplication by $e^{-ik\psi}$ applied to hypothesis weights from the $U$ matrix or, by linearity, the template features themselves. If in-plane rotations are discretely sampled at the same angular spacing rate of the finest HEALPix grid ($\Delta \psi=0.46^\circ$), then our decomposition provides a compressed representation for a matrix with $n=153\times10^6$ rows.

Runtimes for the decomposition pipeline scale approximately linearly with problem size when run on a GPU. The smallest decomposition (768 orientations) ran in 3.4 s on a workstation equipped with an RTX PRO 6000 GPU while the largest decomposition (196,608 orientations) took only 14.2 minutes on the same hardware (Figure~\ref{fig:fig5-fullscale}C). Though the latter exceeds current 2DTM out-of-plane orientation sampling by two orders of magnitude, decomposing a template matching matrix spanning hundreds-of-thousands of rows remains tractable. The ability to decompose a template matching matrix representing sufficient orientational sampling to represent all possible projections of a particle given the Nyquist resolution of the dataset could reduce false negatives that result from missing views in the template matching matrix. Sampling along additional axes, such as molecular conformation, and projecting all hypotheses into the same space using the \emph{symmetry exploiting SVD} is also feasible on workstation-level hardware.

The cumulative unexplained variance ratios from each independently computed decomposition decay at a consistent rate across sampling densities (Figure~\ref{fig:fig5-fullscale}D). The unexplained variance of the two coarsest grids ($n'=768$ and $n'=3072$) begin to deviate from the finest grid near component indices $10^3$ and $10^4$, respectively, where the number of orientation samples begin to bound the achievable matrix rank. Recovered template features are consistent only at low wavenumbers and eigenvector indices for the coarsest sampled grid. As more orientations are sampled for the decomposition, the template features converge for higher wavenumbers and eigenvector indices (Figure~\ref{fig:si-feature-similarity}). Despite extracting different underlying template features, all five decompositions perform equally well at reconstructing projections at the 768 similar orientations shared between all grids (Figure~\ref{fig:fig5-fullscale}E).

We were able to compute the symmetry exploiting SVD for a typical number of out-of-plane orientations ($n'=3072$) for a template matching matrix in \textbf{15.7} seconds, including the projection generation and pre-transformation stage, on an RTX PRO 6000 Workstation Edition GPU 
(Figure~\ref{fig:fig5-fullscale}C). For comparison, computing the projections by enumerating 240 in-plane rotations for each out-of-plane orientation to construct a template matching matrix takes approximately 7 minutes on the same GPU hardware. So, by exploiting symmetry, decomposing a full-scale template matching matrix takes \textbf{25} times less time than computing its entries.

When approximating the action of a template matching matrix, an inverse Fourier transform can be used to recover cross-correlation values along in-plane rotations more efficiently than rote enumeration. Each wavenumber, $k$, of a template feature contributes a term $e^{-ik\psi}$ to the associated hypothesis weight at angle $\psi$. Summing these contributions across features grouped by wavenumber produces a complex-valued frequency spectrum, ${\hat{C}}$, where
\begin{equation}
    {\hat{C}}^{(k)} = M'^{(k)} x = U^{(k)} \Sigma^{(k)} ({W^*}^{(k)}x)
\end{equation}
is the $k^\text{th}$ angular frequency. The in-plane rotation cross-correlation may then be recovered by a single real-valued inverse Fourier transform, $C_\psi = \text{IRFFT}_k\left( {\hat{C}}\right)$. This mirrors prior strategies for fast rotational image alignment, where the best matching rotation is recovered from a 1D Fourier transform over a polar representation of an image \cite{de-castro_1987_registration, reddy_1996_an}.

\section{... for compressed search.} \label{sec: discussion}

Many problems in biology seek to understand how cellular physiology and disease are mediated by their constituent molecular components. Cryo-EM produces atomic resolution views of cells but the potential to visualize the molecular structure and interactions of proteins within their native cellular environment is limited by the difficulty in interpreting these images. Template matching has been successfully applied to annotate specific macromolecular structures within cellular cryo-EM images but cannot be scaled to the proteome because of the extreme computational demands.



The procedure described in this manuscript allows the decomposition of full-scale template matching matrices that were too large to decompose using standard linear algebra packages, even with large-scale computing resources. Leveraging the in-plane rotational symmetry dramatically reduces the computational costs of decomposition, allowing decomposition of template matching matrices in less time than their construction and enabling the decomposition of template matching matrices with finely sampled angular coordinates. 
Symmetries of cryo-EM images related by an in-plane rotation have similarly motivated steerable PCA algorithms built around Fourier-Bessel representations \cite{ponce_2011_computing, zhao_2014_rotationally, zhao_2016_fast}. These algorithms are used to project experimental data onto a low-dimension, orthogonal basis where the primary interest is the covariance structure of the data and where preserving noise statistics is important \cite{marshall_2023_fast}. Our work differs in that we seek a linear decomposition of a deterministic template matching matrix towards approximating the action of the matrix on an image up to a desired accuracy. This objective does not depend on the scale of an experimental dataset or require preservation of noise in the reference. To exploit the rotational symmetry of the template matching matrix during decomposition, a polar representation of projections is sufficient as long as the polar-Cartesian interpolation error is controlled. This removes extra machinery associated with the non-uniform FFT and Bessel-quadrature stages needed for constructing a Fourier-Bessel image representation.


Like any fast numerical routine, the existence of an efficient procedure expands the scope of problems we could hope to solve via decomposition beyond the current template matching pipeline.

\begin{itemize}
    \item \textbf{Compress More Hypotheses:} We could expand the number of hypotheses represented in a single decomposition by adding more projections to the template matching matrix. These could correspond to modeling conformational flexibility of dynamic target molecules or including multiple distinct target structures in the same decomposition.

    In every case, the required decomposition would provide a low rank approximation to the full template matching matrix by projecting its rows onto a series of orthonormal template features. These will always adopt the separable form identified in Section \ref{sec: Fast Diagonalization}. In particular, every template feature will equal an outer product between a angular Fourier mode and a radial amplitude function. 
    
    Recall that, the amplitude functions associated with wavenumber $k$ are the right singular vectors of the matrices $\tilde{M}^{(k)}$. Therefore, the amplitude functions may be recovered from the eigenvectors of $B^{(k)} = \left(\tilde{M}^{(k)}\right)^{*} \tilde{M}^{(k)}$ \cite{strang2012linear}. So, given $h$ continuous hypothesis coordinates, $\{z_j\}_{j=1}^h$, the necessary amplitude functions, associated with wavenumber $k$, are the eigenfunctions of the integral operator with kernel:
    \begin{equation} \label{eqn: integral operator}
        B^{(k)}(\rho,\rho') = \iiint \hat{M}(\rho,k;z)^* \hat{M}(\rho',k;z) dz_1 dz_2 ... dz_h
    \end{equation}
    where $M(\rho,\psi;z)$ is the intensity of the projection corresponding to hypothesis $z$, at $(\rho, \psi)$ in the imaging plane, and where $\hat{M}(\rho,k;z)$ is the Fourier transform of $M(\rho,\psi;z)$ with respect to $\psi$ at wavenumber $k$. This is the approach used to find steerable expansions of individual filters in \cite{perona1995deformable,uenohara1998optimal}. To append additional hypothesis coordinates, simply marginalize over an additional $z_j$. If a hypothesis coordinate is discrete, use a sum instead of an integral. 

    Since the coefficients needed to represent any particular hypothesis can be recovered from a known set of template features by projection onto the template features, we can separate the problem of resolving the template features, $W$, from the hypothesis features, $U$, in the SVD of $M$. When computing $W$, the hypothesis coordinates need only be sampled finely enough to provide accurate quadrature approximations to the integrals in \eqref{eqn: integral operator}. Then, we can recover the entries of $U$ for any new projection by computing its product with $W$. Thus, when the entries of $M$ depend smoothly on a hypothesis coordinate, we may be able to compute the features using a coarse discretization of the hypothesis space, then query the features during a template matching search at a fine discretization. 


    \item \textbf{Multi-Precision Template Matching:} The availability of singular value decompositions of full-scale template matching matrices allows the reconstruction of cross-correlograms with a Pareto optimal trade-off between precision and computational effort. This flexibility will enable multi-precision search strategies that use validated numerics \cite{tucker2011validated,moore2009introduction,nedialkov2004interval} to perform the same statistical search used in 2DTM, while only computing each cross-correlogram entry to its minimal necessary precision. In principle such a strategy has the potential for faster, higher throughput template matching by reducing the overall computational cost. 
    

\end{itemize}

Because the fast diagonalization algorithm (Algorithm \ref{alg: variable size}) applies to any matrix that exhibits a symmetry of action, we expect that the advantages demonstrated in this paper, in the context of template matching, will apply in other contexts that exhibit such a symmetry. 


\section*{Datasets and Code}

Implementations and results discussed in this manuscript for the decomposition of template matching matrices can be found in a Python package located at \href{https://github.com/Lucaslab-Berkeley/Panther-EM}{github.com/Lucaslab-Berkeley/Panther-EM}, specifically \texttt{v0.0.1-alpha}. Analysis scripts and input data available upon request.

\section*{Acknowledgements}

The authors thank Gilbert Strang for his suggested references. They thank Niko Grigorieff and Josh Dickerson for their critical reading and insightful feedback on preliminary manuscripts and Tim Grant for initial conversations about 2DTM using low rank approximations. 

\section*{Funding}
This work was made possible by funding from an NIH Director’s New Innovator Award (DP2GM159184). BAL is a Searle Scholar and a Shurl and Kay Curci Scholar.

\begingroup
\small
\bibliographystyle{pnas-new}
\bibliography{Refs}  
\endgroup

\newpage
\subsection*{Supporting Information Appendix (SI)}

\subsection{Optimality of the Singular Value Decomposition} \label{app: optimality of SVD}

Let $N \in \mathbb{C}^{m \times n}$ denote an un-whitened template matching matrix. A matched filter search begins by computing a matrix vector product of the form $N x$ where $x$ is a vectorized image. We pursue a rank $r$ approximation to $N$, $N \approx N^{(r)}$, to approximate products against $N$ efficiently. Therefore, we pursue a low rank approximation such that the approximation $N^{(r)} x \approx N x$ is accurate for typical images $x$.

Suppose that $x$ is drawn from an ensemble of images with covariance, $\text{Cov}[x] = C$. To whiten the problem, replace $x$ with $y = C^{-1/2} x$. Then $\text{Cov}[y] = C^{-1/2} C C^{-1/2} = I$. Then, to work with $y$, we should replace $N$ with $M = N C^{1/2}$ so that $M y = N x$. Let $M^{(r)} = N^{(r)} C^{-1/2}$ denote the matching rank $r$ approximation to the whitened template matching matrix. 

Let $E = M^{(r)} - M$ denote the error in the rank $r$ approximation to the whitened template matching matrix. The error in the approximate product is $e = (M^{(r)} - M) y = E y$. Then:
\begin{equation}
    \mathbb{E}[\|e\|^2] = \mathbb{E}[y^{\intercal} E^{\intercal} E y] =  \langle E^{\intercal} E, \text{Cov}[y] \rangle = \langle E^{\intercal} E, I \rangle
\end{equation}
where $\langle A, B \rangle = \sum_{i,j} A_{i,j} B_{i,j}$ is the matrix-inner product. Then:
\begin{equation}
    \mathbb{E}[\|e\|^2] = \sum_{i=1}^n \sum_{j=1}^n [E^{\intercal} E]_{i,j} \delta_{i,j} = \sum_{i=1}^n [E^{\intercal} E]_{i,i} = \sum_{i=1}^n \sum_{j=1}^n E_{i,j}^2 = \|E\|_{\text{Fro}}^2
\end{equation}
where $\|A\|_{\text{Fro}}$ denotes the Frobenius norm.

Let $M = U \Sigma W^*$ denote the SVD of $M$. By Eckart-Mirsky-Young \cite{eckart1936approximation,mirsky1960symmetric}, the low rank approximation produced by truncating the SVD at $r$ components:
\begin{equation}
    M \approx M^{(r)} = \sum_{i=1}^r \sigma_r U_{r,:} W_{r,:}^*
\end{equation}
minimizes $\|E\| = \|M^{(r)} - M\|$, for any unitarily invariant matrix norms, among all rank $r$ matrices $M^{(r)}$. The Frobenius norm is a unitarily invariant matrix norm, so the truncated SVD minimizes $\|E\|_{\text{Fro}}^2$ among all rank $r$ approximations to $M$. 

Therefore, the truncated SVD to the whitened template matrix provides the best rank $r$ approximation to $M$ in the sense that it minimizes the expected square error in the products used in the first stage of a matched filter search.

\subsection{Fast Diagonalization Algorithm for Generic Symmetries} \label{app: general algorithm}

If $M$ is a square matrix that commutes with a generic permutation, then it may be diagonalized efficiently using the procedure outlined in Section \ref{sec: Fast Diagonalization}. That procedure is enumerated as an algorithm in Algorithm \ref{alg: variable size}.

\begin{algorithm}
\caption{Symmetry Exploiting Diagonalization (Generic)}
\label{alg: variable size}

Given $M \in \mathbb{C}^{n \times n}$, diagonalizable, block-circulant with $l$ blocks of size $\{m_j\}_{j=1}^l$:

\begin{algorithmic}[1] 
\State Order the indices lexicographically with outer indices corresponding to equivalence class, $E(x)$, and inner indices corresponding to position within an equivalence class, $j(x)$.
\State $\hat{M}^{(i,j)} \gets \bar{F}_{m_i} M^{(i,j)} F_{m_j}$ performed implicitly, in parallel, by applying the DFT to the first $p_{ij} = \text{gcd}(m_i,m_j)$ entries of the first row of $M^{(i,j)}$
\State List all frequencies, then sort in increasing order. Set $S$ equal to the corresponding permutation matrix.
\State  Identify each unique frequency in the list then let $\tilde{M}^{(k)}$ denote all entries of $\hat{M}$ whose row and column indices correspond to the $k^{th}$ unique frequency.
\State $(A^{(k)},\Lambda^{(k)}) \gets \text{eig}(\tilde{M}^{(k)})$ for each group $k$. Perform in parallel. 
\end{algorithmic}

\vspace{0.02 in}
Return: $(S, A, \Lambda)$ and $V =  F S A$. Then $(\Lambda, V)$ provide the eigendecomposition $M = V \Lambda V^{-1}$. 
\end{algorithm}

\begin{remark} \label{remark: costs generic} Suppose that $M \in \mathbb{C}^{n \times n}$ commutes with a permutation $\sigma$ which divides $\mathcal{X}$ into $l$ equivalence classes of size $\{m_j\}_{j=1}^l$. Let $p_{ij} = \text{gcd}(m_i,m_j)$. 

\begin{enumerate} [leftmargin=0.5cm]
    \item \textbf{Storage:} The memory needed to store $M$, diagonalize $M$, and store its eigenvectors is: $\mathcal{O}(\sum_{i,j=1}^l p_{ij})$

   \item \textbf{Computation:}  The computation needed to diagonalize $M$ is $\mathcal{O}(\sum_{i,j=1}^l p_{ij} \log(p_{ij})/w) + \mathcal{O}((n/w) \log(n)) + \mathcal{O}(\sum_{k} (b_k - 1)^3/w)$ where $b_k$ is the size of the $k^{th}$ block of $\tilde{M}$ produced after reordering $\hat{M}$. The $\mathcal{O}((n/w) \log(n))$ term represents the cost to find frequency groups. This cost may be distributed across repeated decompositions of different matrices with the same symmetries. 
 
\end{enumerate}
\end{remark}

\newpage
\subsection{Decomposition Pipeline for Template Matching Matrices} \label{app: decomp pipeline}

\begin{figure}[h!]
    \centering
    \includegraphics[width=0.38\linewidth]{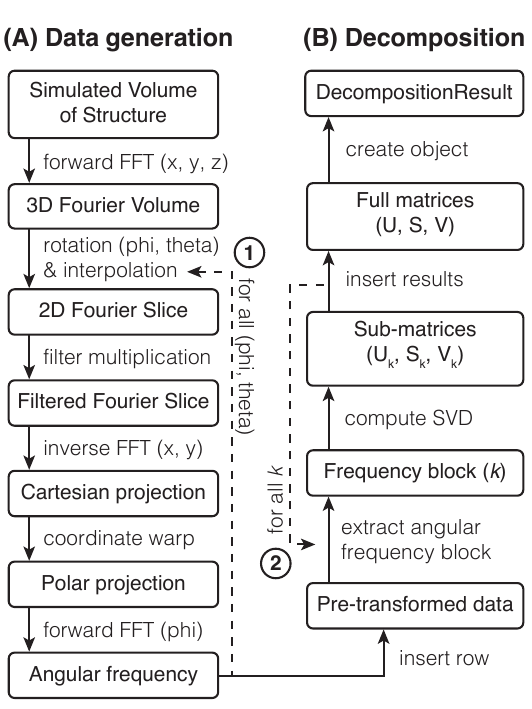}
    \caption{Discrete steps along the decomposition pipeline. (A) Projection images are generated on-the-fly via the Fourier slice theorem. Cartesian projections are warped into polar coordinates using bi-quintic interpolation and pre-transformed into an angular frequency representation. All out-of-plane angles $(\phi, \theta)$ are processed before the decomposition stage. (B) Each angular frequency block is processed individually to create a full SVD representation of the data. Results are packaged into a helper \texttt{DecompositionResult} Python object. Dashed lines are trivially parallelizable stages in the pipeline.}
    \label{fig:fig6-workflow}
\end{figure}

\newpage
\subsection{Spiral polar coordinate system} \label{app: spiral}

The standard, discretized polar coordinate system has a uniform number of radial nodes ($n_\rho$) equally spaced from zero radius to the maximum considered radius ($\rho_\text{max}$) and each containing the same number of uniformly spaced angular samples ($n_\psi$). Adopting $t$ as a parameterization variable, the standard polar coordinate system for coordinates $(\rho, \psi)$ is defined as
\begin{equation}
    \rho(t)=\rho_\text{max}t, \quad \psi(a) = \frac{2\pi}{n_\psi}a,
\end{equation}
with $t=i/(n_\rho -1)\in[0, 1]$, discretized radial index $i=0,\dots,n_\rho-1$, and angular index $a=0,\dots,n_\psi-1$.

The spacing between rays in the polar grid increases with radius. As a result, the point density is non-uniform, and, Cartesian grid points far from the origin may be far from the nearest grid point in the polar grid. These features produce large interpolation errors when the original image contains high frequency angular information at large radii since interpolation onto the polar grid, then back to the Cartesian grid, smooths angularly where the polar grid is less dense. Custom polar coordinate systems that use irregularly spaced rings were also proposed in \cite{hilai1994recognition}.

We leveraged a polar-like coordinate system where spacing between subsequent radial nodes grows with the square root of distance and where each ring has a radius-dependent angular offset. This coordinate system still obeys rotational symmetry in each radial node since it is built from $n_{\psi}$ copies of a single spiral arm extending from the origin, rotated incrementally by $\Delta \psi$. Hence, we refer to this coordinate system as spiral polar coordinates.

Unlike the standard polar grid, a spiral grid can achieve asymptotically uniform point density as a function of $\rho$ by reducing the space between each ring of grid points as the radius increases. By choosing the spacing to decay at rate $\mathcal{O}(\rho^{-1})$, we ensure that the area of each triangle in a triangulation of the grid approaches a constant as $\rho$ increases. By offsetting the point placement in adjacent rings, we ensure that the number of spiral grid points per cell of the original Cartesian grid is approximately constant at large radii. 

Using the same indexing as above, the spiral polar coordinate system is defined as
\begin{equation}
    \rho_\text{spiral}(t) = \rho_\text{max}\left(\sqrt{c^2 + t(1+2c)} - c\right), \quad
    \psi_\text{spiral}(a,t) = \frac{2\pi}{n_\psi}\left(a + p_0\, n_{\psi}\, t\right),
\end{equation}
where $c\geq0$ controls where the grid changes from uniform to square root radial spacing and where $p_0$, in radians, controls the percent arc offset between subsequent radial nodes. 

The square root growth rate is chosen to achieve asymptotically uniform point density, while the offset $c$ ensures that the radial spacing is approximately linear near the origin. Increasing $c$ increases the number of rings near the origin. The parameter $c$ can be chosen to bound the worst case distance between a point inside the disc and any point in the mesh. Controlling both the worst case distance (maximum radius of any triangle in the triangulation) and the areas of the triangles, ensures a well-conditioned mesh with an approximately equilateral triangulation. This minimizes interpolation artifacts and reduces smoothing when moving to and from a Cartesian grid.

\begin{figure}[h!]
    \centering
    \includegraphics[width=0.7\linewidth]{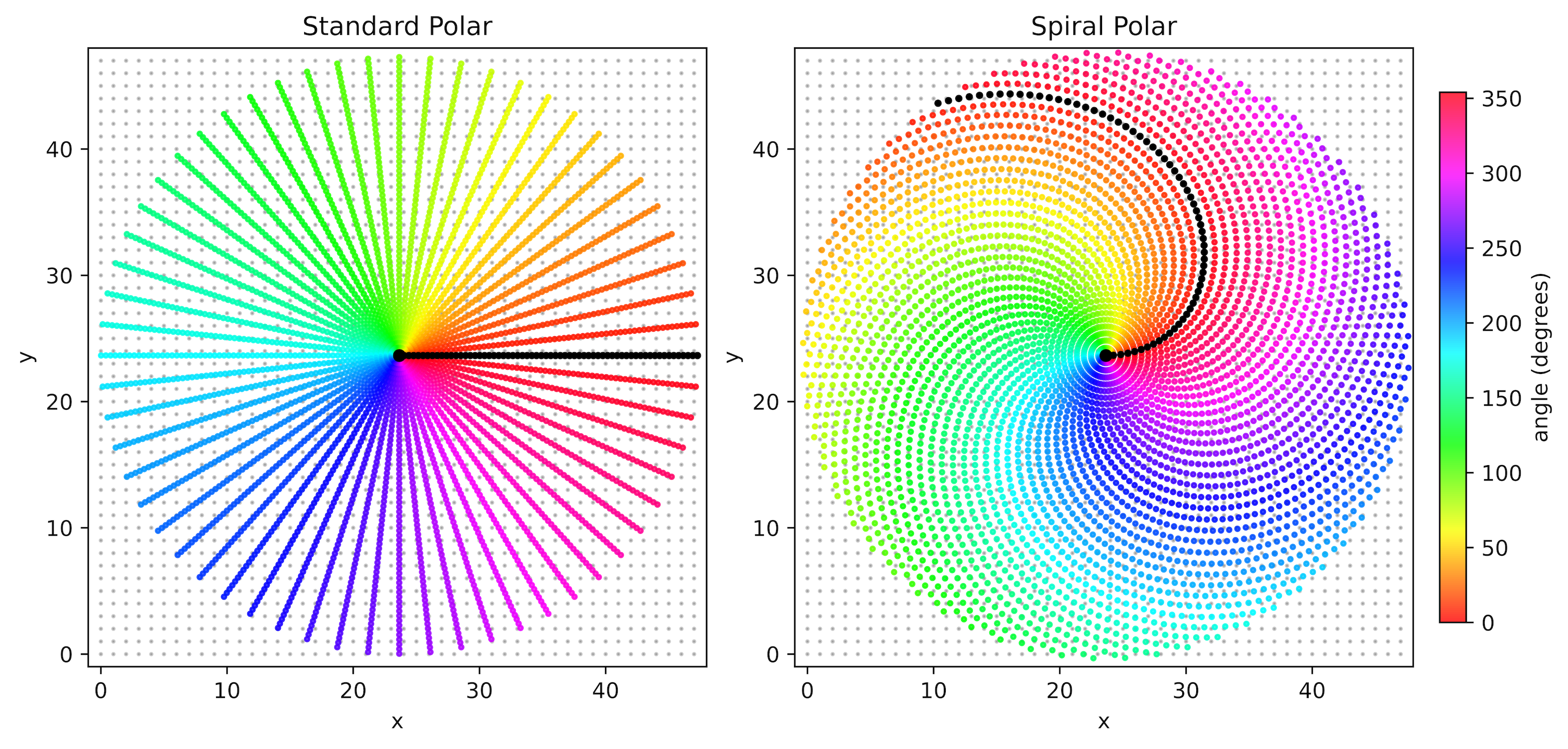}
    \caption{Visual representation of the standard polar and spiral polar coordinate systems. At large radii, the spiral polar coordinate system has nodes closer to the underlying Cartesian grid. When compared to the standard polar grid with the same number of radial and angular nodes, the spiral polar coordinate system introduces less smoothing through round-trip interpolation. Both grids have $n_\rho=64$ radial nodes and $n_\psi=60$ angular nodes}
    \label{fig:spiral-polar}
\end{figure}

When optimized to minimize round-trip interpolation error on simulated $512\times512$ pixel projections of the ribosome, we found that $n_r=1024$, $n_\psi=1200$, $c=1.0$, and $p_0=0.2$ performed best with an average error of $0.127\%$.




\newpage
\subsection{Convergence of template features on progressive HEALPix grids}

\begin{figure}[h!]
    \centering
    \includegraphics[width=1\linewidth]{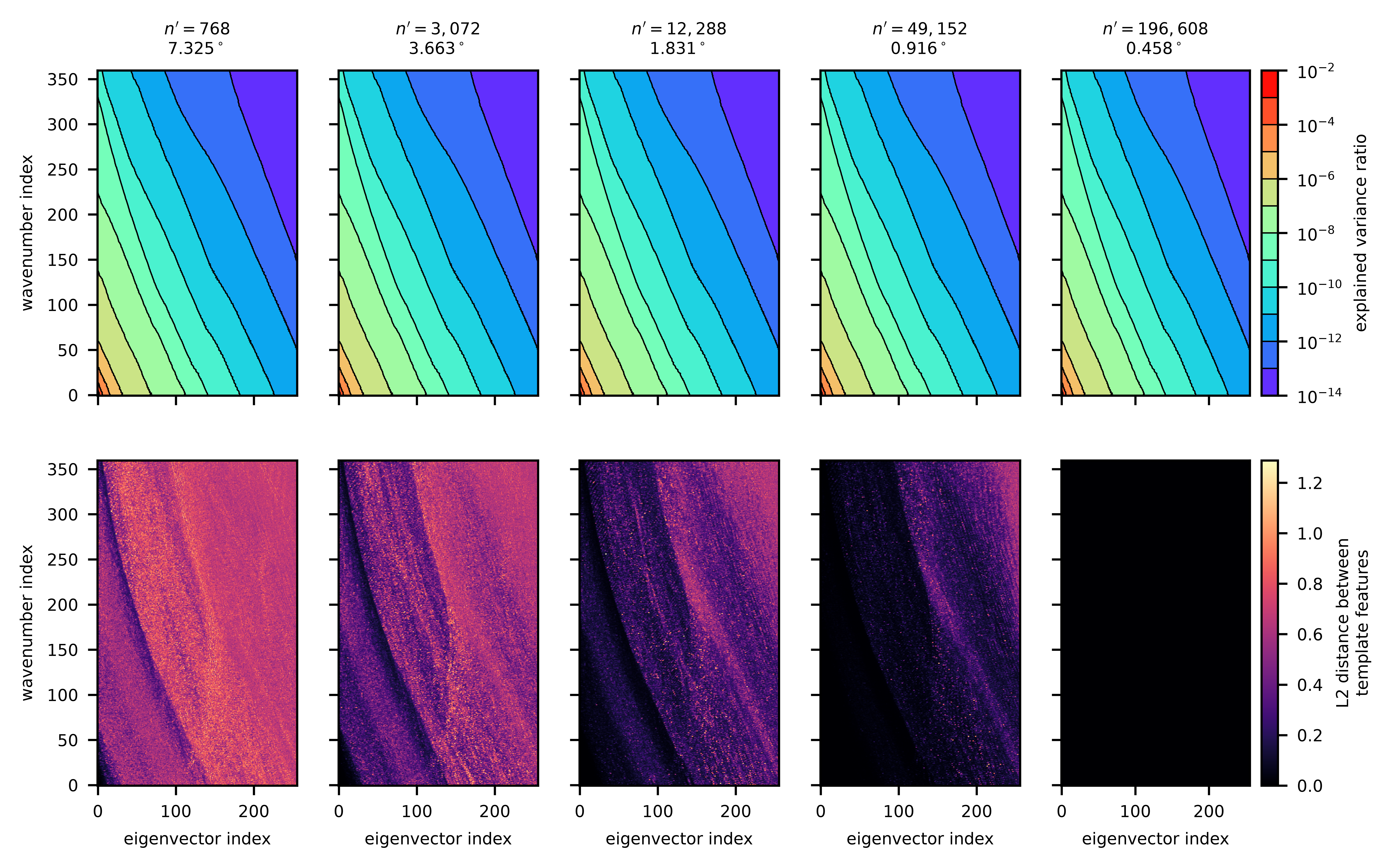}
    \caption{Comparison of singular values and template features from symmetry exploiting SVD of large ribosomal subunit projections sampled across the 5 HEALPix grids for out-of-plane orientation. (Top) Explained variance ratios for each singular value in the wavenumber and eigenvector index grid. Despite the orientation sampling spanning nearly 3 orders of magnitude, the explained variance ratios display nearly exact correspondence since each orientation sampling approximates the same underlying integral operation. (Bottom) L2 distances between template features (radial amplitude functions) extracted from each decomposition compared to the finest sampled grid. Very low wavenumber and eigenvector indices converge to the same template features even at the coarsest sampled grid. As the number of orientations increases, higher index template features converge. This observation is consistent with increased orientation sampling the underlying integral operator to a higher degree.}
    \label{fig:si-feature-similarity}
\end{figure}

\newpage
\subsection{Expected computational advantage of an incremental, multi-precision approach}

Let $X$ be a square image with height and width $H$, and let a square simulated projection (row of template matching matrix $M$) have height and width $h \ll H$. Computing the cross-correlation between the simulated projection and $X$ using the convolution theorem and FFT can be done in $\mathcal{O}(H^2\log H)$ work. One must also perform a per-pixel reduction after each cross-correlation calculation to find where the current hypothesis better than any of the previous hypotheses which can be done in $\mathcal{O}(H^2)$ work. When iterating over the $n$ hypothesis in $M$ row-by-row to find the best matching hypothesis per-location in the image, the cost of the FFT-based 2DTM search algorithm scales as:
\begin{equation}
C_\text{fft}=\mathcal{O}(n H^2\log H)
\end{equation}

A 2DTM search leveraging a symmetry exploiting SVD decomposition treats in-plane rotation separately from other hypotheses in the search such that the total number hypotheses, $n$, are a product of the non in-plane rotation hypotheses and the number of in-plane rotations $n=n'\cdot n_\psi$. For some selected $r_0$ components, a search begins with image featurization with weighted template features,
\begin{equation}
    Z=\left( \Sigma_{r_0} W_{r_0}^* \right) \star X
\end{equation}
where $\star$ is the spatial cross-correlation operator. The featurization stage, by leveraging the convolution theorem and FFT, can be done in $\mathcal{O}(r_0 H^2\log H)$ total work.

Recovery of cross-correlation values from the featurized image proceeds as $n'$ inner products per-pixel, one for each non in-plane rotation hypothesis, followed by an inverse real-valued FFT to recover in-plane rotations. Computation of the inner products scales as $\mathcal{O}(r_0 n')$, and in-plane rotation recovery scales as $\mathcal{O}(n_\psi \log n_\psi)$. Applying these computations across all $H^2$ pixels in the featruized image, assuming constant cost from the reduction stage, has the total SVD-based search cost scales as:
\begin{equation}
    C_\text{svd}=\mathcal{O}(r_0H^2 \log H) + \mathcal{O}(H^2(r_0 n' + n_\psi \log n_\psi))
\end{equation}

A low-rank SVD-based approximation for a 2DTM search with a single stage becomes computationally advantageous when $C_\text{svd}<C_\text{fft}$. Assuming $n'\gg n_\psi \log n_\psi$ and a constant cost per feature, this advantage occurs when $r_0 < n_\psi \log H$. Replacing the per-rotation search with a low-rank inner product is advantageous when the retained rank, $r_0$, falls below $n_\psi \log H$ which can be thought of as the effective per-hypothesis cost of a rotation search by direct correlation.





\restoregeometry
\end{document}